\documentclass[prc,aps,nofootinbib,showkeys,showpacs,twocolumn]{revtex4}
\usepackage{epsfig}
\usepackage{graphicx}
\usepackage{amssymb}
\usepackage{color}
\usepackage{upgreek}
\usepackage{mathtools}
\usepackage{comment} 
\usepackage{tcolorbox}
\usepackage{amsmath, amsthm, amssymb}

\newtheorem{proposition}{Proposition}

\usepackage{braket}
\usepackage{comment}
\usepackage{ulem}
\usepackage{bbm}

\usepackage{makecell}

\begin{document}

\title{Optimizing quantum encodings for analog simulation through dynamical algebra reachability}

\author{Mariane Mangin-Brinet} \email{mariane@lpsc.in2p3.fr}
\affiliation{Laboratoire de Physique Subatomique et de Cosmologie, CNRS/IN2P3/UGA, 38026 Grenoble, France}

\date{\today}
\begin{abstract}
Analog quantum computers  provides direct access to continuous many-body dynamics, but the native control Hamiltonians of an analog processor generate only a restricted operator space. Consequently, the fidelity with which they can reproduce a target Hamiltonian's dynamics depends not only on spectral agreement but on 
whether the physical controls can actually generate the required evolution, or, said differently,  on whether the encoded dynamics is compatible with the directions accessible to the hardware. 
We introduce a geometry-independent framework for diagnosing and optimizing this compatibility through the principal angles between the spectral algebra $\mathcal A(H)$ generated by a target Hamiltonian $H$ and a Krylov-type operator space $\mathcal K$ generated from the device's independently tunable control Hamiltonians and a fixed initial state. Because unitary conjugation preserves the spectrum of $H$, the encoding problem can be formulated as an optimization over the unitary orbit of the target Hamiltonian. We maximize a smooth subspace-overlap functional using Riemannian gradient descent on $U(d)$, thereby selecting a spectrally equivalent representation whose target algebra is better aligned with the native controls.
We apply this framework to the Hamming-weight-one deuteron Hamiltonian encoded on a Rydberg-atom analog processor. For global Rabi, phase, and detuning controls of a $N$ atoms computer, the geometry-independent Krylov space has dimension $\binom{N+3}{3}-1$. The standard binary encoding is strongly misaligned with this space for small system sizes, while 
principal-angle optimization substantially improves the algebraic compatibility of the encoding. For $N=3,4$, it can place the full target spectral algebra within the accessible operator 
space, whereas for $N=5,\ldots,8$ only partial algebraic alignment is achieved. Nevertheless, high-fidelity state preparation remains possible for the latter systems when geometry-dependent 
control resources are optimized. Moreover, using a single geometry-independent optimized encoding substantially reduces the sensitivity of state-preparation fidelity to the atomic geometry. 
These results establish algebraic reachability as a useful preprocessing criterion for analog encoding design: it identifies representation-level incompatibilities before device-specific geometry and pulse optimization, while the distinction between algebra-level and state-level reachability clarifies when full encoding optimization is necessary and when geometry-dependent control resources can compensate for incomplete algebraic alignment.
\end{abstract}


\maketitle

\section{Introduction}

Analog quantum computers offer a fundamentally different paradigm for quantum simulation: instead of decomposing an algorithm into a sequence of elementary digital gates, they exploit the continuous evolution generated by a physical system to directly emulate the dynamics of a target model. This approach has attracted renewed interest for the study of many-body phenomena, lattice gauge theories, and nuclear structure problems on engineered platforms such as neutral-atom arrays, trapped ions, and superconducting circuits, where effective Hamiltonians with controllable interactions can be realized natively~\cite{Surace2020,Zhu2024,Morgado2011,Georgescu2014}. However, an analog device does not implement an arbitrary target Hamiltonian; it implements whatever Hamiltonians its native driving and interaction terms can generate. The central question for analog simulation is therefore not only whether a target model can be approximated energetically, but whether the encoding chosen to map the target problem onto the device's qubits preserves the dynamical structure that the device is actually capable of generating.

This question becomes acute once one recognizes that a given physical problem can be mapped onto a quantum register in many inequivalent ways. Any two encodings related by a unitary change of basis describe exactly the same spectrum and the same physics, yet they can differ enormously in how compatible they are with a specific set of native controls. A poorly chosen encoding may place the relevant dynamics along directions of Hilbert space that the analog controls simply cannot reach, regardless of how the pulse sequence, evolution time, or device geometry are optimized. Conversely, an encoding that is well matched to the native control algebra can make an otherwise inaccessible target state fully reachable. Distinguishing these two situations, and turning the distinction into a constructive design principle, is the goal of this work.

We address this question through the lens of dynamical algebra reachability. In control theory, reachability characterizes which states or transformations can be generated from a given set of controls~\cite{Zeier2011,DAlessandro2020,DAlessandro2021}; here we extend this notion from individual states to the full algebraic structure of the target Hamiltonian, in the spirit of recent classifications of the
dynamical Lie algebras generated by spin Hamiltonians~\cite{Wiersema2024}. We characterize a target model not by its Hamiltonian alone but by the spectral algebra $\mathcal A(H)=\operatorname{span}\{I,H,H^2,\ldots,H^{r-1}\}$ that it generates, and we characterize the analog device by the operator space $\mathcal K$ reachable from a fixed initial state under repeated commutators with its native control Hamiltonians -- the analog-control counterpart of a Krylov space. The compatibility between a given encoding of the target model and a given analog platform is then quantified by the principal angles between $\mathcal A(H)$ and $\mathcal K$: complete controllability of the encoded algebra corresponds to all principal angles vanishing, while any nonzero angle signals a direction of the target dynamics that the native controls cannot access.

Because $\mathcal A(H)$ is transported rigidly under unitary conjugation of $H$, while $\mathcal K$ is fixed by the hardware and the initial state, encoding optimization can be posed as a well-defined variational problem: find the unitary representation of the target Hamiltonian whose transported spectral algebra best aligns with the reachable operator space. We formulate this as a smooth optimization over the unitary orbit of $H$, using a subspace-overlap functional built from the singular values of the algebra-Krylov overlap matrix, and solve it by Riemannian gradient descent with backtracking line search directly on the unitary manifold, in the
same spirit as gradient-based optimal-control methods used elsewhere in quantum control~\cite{Khaneja2005}. Unlike optimization over a discrete family of encodings (standard binary, Gray code, permutation, or Dicke-state mappings), this approach explores the complete space of physically equivalent Hamiltonian representations and returns, by construction, the encoding that is optimal with respect to the chosen algebraic criterion.

We apply this framework to a benchmark problem of direct relevance to nuclear physics~\cite{BauerNatRevPhys2023,BauerPRXQuantum2023}: the deuteron Hamiltonian of Dumitrescu et al.~\cite{Dimitrescu2018}, restricted to its physical Hamming-weight-one sector and encoded onto a Rydberg-atom analog processor. The sensitivity of nuclear quantum-simulation results to the choice of fermion-to-qubit encoding has already been noted in a purely digital, gate-based setting~\cite{Siwach2021}; the present work identifies an analogous, but qualitatively distinct,
encoding sensitivity that is specific to analog control. 
The native Rydberg controls -- global Rabi drive, detuning, and phase, together with the geometry-dependent Ising-type interaction -- generate a reachable operator space whose 
structure we analyze first in a geometry-independent setting. For the global controls alone, this space is permutation symmetric, while the specific atomic geometry generally breaks 
this symmetry through the pair-dependent interaction strengths. We derive the dimension and structure of the geometry-independent Krylov space in closed form and subsequently 
examine how geometry-dependent interactions modify state-level reachability.

Comparing this geometry-independent Krylov space against the spectral algebra of the standard binary (SB) encoding for successive truncations $N = 3$ through 8, and cross-checking the 
algebraic prediction against direct optimal-control simulation of the Rydberg evolution, we find a sharp and consistent pattern at the level of state reachability, even though algebra-level 
containment remains incomplete throughout: for $N = 3$ and $N = 4$, the SB encoding is so poorly aligned with the native controls -- with only a one-dimensional overlap between a 
three- or four-dimensional target algebra and the reachable Krylov space -- that no amount of pulse or geometry optimization can prepare the target ground state to high fidelity. 
For $N \ge 5$, the algebra-level overlap remains similarly one-dimensional, and full algebraic containment is never achieved for any $N$ in the range studied; nevertheless, once 
the atomic geometry is included as a free control parameter, joint optimization of the geometry and pulse sequence succeeds in preparing the ground state to fidelities exceeding 
0.999, without any change to the encoding itself. Applying the proposed unitary-orbit optimization to the $N = 3, 4$ cases removes the underlying algebraic obstruction entirely, driving 
every principal angle cosine to unity and recovering exact ground-state reachability at fixed geometry, while leaving the physical spectrum of the target Hamiltonian untouched.

This last observation -- that full containment of the encoded algebra within the geometry-independent Krylov space is sufficient, but in practice not necessary, for high-fidelity preparation of a specific target eigenstate --  points to a broader distinction between two levels of reachability that this work also aims to clarify. Algebra-level reachability guarantees access to any observable or eigenstate generated by the target Hamiltonian, and is essential when the state of interest is not known in advance or when full spectral control is required. State-level reachability, targeting only a single eigenvector such as the ground state, is a substantially weaker and cheaper requirement, and we show that it can already be satisfied by an unoptimized encoding once the reachable operator space is large enough, even when the encoded algebra as a whole remains only partially controllable. Identifying which of these two regimes a given problem falls into is, in itself, a useful diagnostic for deciding whether the cost of encoding optimization is justified.

The remainder of the paper is organized as follows. Section~\ref{sec:UOOpt} formulates the unitary-orbit optimization, defines the spectral algebra and the geometry-independent Krylov space, 
introduces the principal-angle objective, and describes the Riemannian optimization procedure. Section~\ref{sec:ControlH} derives the native Rydberg control Hamiltonians and the closed-form 
structure of the accessible Krylov space. Section~\ref{sec:deuteron} introduces the deuteron Hamiltonian, its Hamming-weight-one reduction, and the scaling of the computational procedure. 
Section~\ref{sec:numerical} presents the numerical results, including the standard-binary benchmark, principal-angle encoding optimization, fixed-geometry tests, and explicit Rydberg control protocols. Section~VI summarizes the conclusions and discusses extensions to larger systems and other analog platforms.

\section{Optimization over the unitary orbit}
\label{sec:UOOpt}

Rather than restricting the search to a finite family of encodings (e.g. standard binary, Gray code, permutation encodings or Dicke-state mappings), we formulate the encoding problem as a continuous optimization over the full unitary orbit of the target Hamiltonian.

Let
\[
\mathcal{A}(H)=\operatorname{span}\{I,H,H^2,\ldots,H^{r-1}\}
\]
denote the spectral algebra generated by the target Hamiltonian, where $r$ is chosen large enough that the algebra has stabilized. 
We note that, for a $d\times d$  Hamiltonian, $r$ is at most $d$ (according to the Cayley-Hamilton theorem)\footnote{For Hermitian Hamiltonians, the spectral algebra is generated by
\[
\{I,H,\ldots,H^{r-1}\},
\]
where \(r\) is the number of distinct eigenvalues of \(H\). In particular,
\[
\dim\mathcal A(H)=r\le d,
\]
with equality if and only if the spectrum is non-degenerate.}.

The second subspace we consider is the operator space accessible to the analog device. Instead of considering the full Hamiltonian, we exploit its independently controllable terms,
\[
H(t)=\sum_k c_k(t)H_k,
\]
where the amplitudes $c_k(t)$ constitute the available controls. 
Starting from an initial state whose density matrix is denoted $\rho_0$, we define the Krylov space generated by the adjoint action of $H$ as
\[
{\cal K}_m(\rho_0)
=
{\rm span}
\Bigl\{
\rho_0,\,
[H_i,\rho_0],\,
[H_i,[H_j,\rho_0]],\,
\ldots
\Bigr\},
\]
where all ordered sequences of controllable terms are included.

Unlike the conventional Krylov space generated by repeated commutators with a single Hamiltonian, ${\cal K}_m(\rho_0)$ incorporates all independent control directions available to the analog device. It therefore provides a state-dependent description of the operator directions that are reachable under the available controls. Equivalently, ${\cal K}_m(\rho_0)$ can be interpreted 
as the action of the dynamical Lie algebra generated by $\{H_k\}$ on the initial state, truncated at commutator depth $m$.

Since the initial state $\rho_0$ is fixed throughout, we shall write $\mathcal{K}$ for $\mathcal{K}(\rho_0)$ when no ambiguity can arise.
$\mathcal{K}(\rho_0)$ is thus fixed once the hardware controls and the initial state have been specified.

The optimization problem then consists in finding a unitary representation of the target Hamiltonian whose spectral algebra exhibits the largest possible overlap with ${\cal K}$. Since every Hamiltonian on the unitary orbit
\[
H(U)=UHU^\dagger,
\qquad
U\in U(d),
\]
is unitarily equivalent to $H$, all such Hamiltonians possess exactly the same spectrum and therefore describe the same physical problem. Under this transformation, the spectral algebra is transported according to
\[
\mathcal A(H(U))
=
U\,\mathcal A(H)\,U^\dagger,
\]
while the controllability space ${\cal K}$ remains unchanged.

The objective is therefore to determine the unitary transformation \(U\) that maximizes the alignment between the encoded Hamiltonian algebra
$
\mathcal A(H(U))
$
and the reachable operator space
$
\mathcal K .
$

The similarity between the transported spectral algebra and the controllability algebra is quantified through their principal angles. Let
\[
Q_A=\{A_i \Big|\; i=1\ldots r\}
\]
and
\[
Q_K=\{K_j \Big|\;  j=1\ldots m\}
\]
be Hilbert-Schmidt orthonormal bases of the two operator spaces. 
We remind that Hilbert-Schmidt inner product is defined as:
\[ \langle A,B\rangle_{HS}=\mbox{Tr}[A^{\dagger}B] \]
Because unitary conjugation preserves Hilbert-Schmidt inner products, an orthonormal basis of the initial spectral algebra can be computed only once and transported throughout the optimization, avoiding repeated orthonormalizations. For a given unitary \(U\), the transported basis is indeed
\[
A_i(U)=UA_iU^\dagger.
\]

The overlap matrix is then
\[
M_{ij}(U)
=
\langle K_i,A_j(U)\rangle_{\mathrm{HS}},
\]

Let \(\{\sigma_i\}\) denote the singular values of this overlap matrix. These singular values are related to the principal angles
\(\theta_i\) between the two subspaces through
\[
\sigma_i=\cos\theta_i .
\]

Full compatibility between the encoded Hamiltonian algebra
$\mathcal{A}(H(U))$
and the reachable algebra generated by the available control Hamiltonians requires not only maximizing the overlap between the two operator spaces, 
but aligning all directions of the encoded algebra with the controllable subspace.

Complete controllability requires that all principal angles between these two subspaces vanish,
\[
\theta_i(U)\rightarrow0,\qquad \forall i ,
\]
or equivalently that all singular values of their overlap matrix satisfy
\[
\sigma_i(U)\rightarrow1.
\]
We therefore minimize the subspace-overlap functional
\[
F(U)=-\sum_i\sigma_i^2 ,
\]
which maximizes the total squared cosines of the principal angles and promotes alignment of the entire encoded algebra with the accessible operator space.

The latter formulation naturally leads to a smooth optimization problem on the unitary manifold. The principal angles optimization is solved in the following using a Riemannian gradient descent on the unitary orbit.
The optimization is thus performed over the unitary orbit of \(H\), searching for the representation whose spectral algebra is the closest to the controllability algebra.

Contrary to objectives based on large powers
\(\sum_i\sigma_i^{2p}\) with \(p\gg1\), which preferentially count only directions already close to perfect alignment and therefore maximize the dimension of the intersection
\[
\dim\left(\mathcal{A}(H(U))\cap\mathcal{K}\right),
\]
the present objective penalizes misalignment of all algebraic directions and is therefore suited for achieving complete controllability.

For a stricter worst-case criterion, one may further maximize the smallest singular value,
\[
F_{\rm min}(U)=-\sigma_{\min}^2 ,
\]
which directly minimizes the largest principal angle and guarantees that no direction of the encoded algebra remains poorly aligned with the reachable space.
This min-type objective is however non-smooth at points where the smallest singular value is degenerate, which complicates gradient-based optimization on the unitary manifold~\cite{Absil2008,SchulteHerbrueggen2008,Wiersema2023} and typically 
requires subgradient or smoothed-minimum techniques. In the applications of Sec.~\ref{sec:numerical}, the sum-based objective 
$F(U)$  already succeeds in driving every principal angle cosine to unity simultaneously, so the additional complexity of optimizing $F_{\rm min}(U)$
 is not required for the cases considered here; we introduce it for completeness and leave a systematic comparison between the two objectives -- in 
 particular for larger systems where full alignment of  $F(U)$ may not be achievable and a worst-case guarantee could be preferable --  to future work.

We optimize the objective on the unitary group using Riemannian steepest descent with a backtracking line search~\cite{Absil2008,SchulteHerbrueggen2008}. At each iteration, the tangent-space gradient is represented by its Hermitian generator \(G\), expanded in an orthonormal basis of Hermitian operators, and the unitary is updated through the corresponding exponential map.
\[
U(t)=e^{-itG}U,
\qquad
G=G^\dagger,
\]
which are evaluated using centered finite differences on the manifold. Since every trial point is obtained through the matrix exponential, unitarity is preserved exactly throughout the optimization.

To guarantee monotonic convergence, each iteration employs a Riemannian backtracking line search. Starting from an initial step size \(\eta_0\), candidate updates
\[
U_{\mathrm{trial}}
=
e^{-i\eta G}U
\]
are generated until the objective decreases. Whenever the decrease is insufficient, the step size is reduced according to
\[
\eta\leftarrow\beta\eta,
\qquad
0<\beta<1,
\]
until an acceptable step is found. This strategy ensures that every accepted iterate remains on the unitary manifold while producing a monotonic decrease of the objective.

The resulting algorithm constitutes a purely mathematical optimization over the unitary orbit of the Hamiltonian. Unlike approaches restricted to specific encoding families, it explores the complete space of equivalent Hamiltonian representations and therefore provides the best possible encoding according to the chosen algebraic criterion.

We remark that although the intersection dimension $\dim\!\left(\mathcal A(H(U))\cap\mathcal K\right)$ is not used as the optimization objective, it provides a natural diagnostic to assess the success of the optimized encoding. Indeed, full controllability of the encoded Hamiltonian algebra is achieved when the entire encoded algebra is contained in the reachable operator space. This ideal condition can be expressed as
\[
\dim\!\left(\mathcal A(H(U))\cap\mathcal K\right)
=
\dim\mathcal A(H(U)).
\]
Equivalently, all directions of the encoded algebra have a non-zero representation within the reachable space, and, when the two spaces have equal dimension, all principal angles vanish. Thus, the intersection dimension serves as a final controllability certificate: it quantifies whether the optimized encoding has made the complete Hamiltonian algebra accessible. However, using it directly as an optimization objective is not desirable, since it is a discrete quantity and may favor encodings with a large but incomplete reachable subspace. Instead, the optimization is performed using a smooth principal-angle-based overlap functional, while the intersection dimension is evaluated a posteriori as a measure of complete controllability.

\section{Control Hamiltonian and accessible Krylov space}
\label{sec:ControlH}

In analog computers, the qubits evolve under a tailored Hamiltonian. We concentrate in this section on Rydberg neutral atom computers, keeping in mind that the procedure described in the previous section 
can be applied to any analog machine. For Rydberg states, the total Hamiltonian is the sum of a driving term and an interaction term:
 \begin{eqnarray*}
 H=\sum_i \Big( H^D_i+\sum_{i<j}H^{int}_{ij} \Big)
\end{eqnarray*}
with 
 \begin{eqnarray*}
 \frac{H^D_i}{\hbar}&=& \frac{\Omega(t)}{2}e^{-i\phi} |g\rangle_i\langle r|_i + \frac{\Omega(t)}{2}e^{i\phi} |r\rangle_i\langle g|_i-\delta(t) |r\rangle_i\langle r|_i 
 \end{eqnarray*}
 where $|g\rangle_i$ and $|r\rangle_i$ denote respectively ground and Rydberg states of atom $i$. The driving Hamiltonian describes the effect of a pulse on two energy levels of an individual atom.
 To keep our applications as close as possible to current analog machines, we will consider global Rabi frequency $\Omega(t)$, detuning $\delta(t)$ and phase $\phi(t)$. 
 
The interaction Hamiltonian depends on the states involved in the sequence. The most common  interaction choice in neutral-atom devices is the Ising Hamiltonian given by
 \begin{eqnarray*}
 \frac{H^{int}_{ij}}{\hbar}&=& \frac{C_6}{R_{ij}^6} \hat{n}_i \hat{n}_j =\frac{C_6}{R_{ij}^6}  (|r\rangle\langle r|)_i  (|r\rangle\langle r|)_j
 \end{eqnarray*}
which depends  on the distance $R_{ij}$ between the atoms $i$ and $j$, and encodes the Ising-like interaction.  $C_6$ is the Ising interaction coefficient that depends on the specific Rydberg state 
$|r \rangle$.

 In terms of Pauli matrices, the total Hamiltonian reads:
\begin{eqnarray*}
 \frac{H}{\hbar}&=&\sum_{i=1}^N \Big(\frac{\Omega(t)}{2}\cos\phi \;X_i - \frac{\Omega(t)}{2}\sin{\phi} \; Y_i -\frac{\delta(t)}{2}(\mathbbm{1}+ \;Z_i))\\
 &+&\sum_{i<j}^N\frac{C_6}{4R_{ij}^6}  (1+Z_i)(1+Z_j)
\end{eqnarray*}
Terms proportional to identity can be dropped since they induce only a phase shift of the state. Terms in $Z_i$ in the interaction part can be re-absorbed into the detuning $\delta(t)$:
\begin{eqnarray*}
&-&\sum_{i=1}^N \frac{\delta(t)}{2} \;Z_i+\sum_{i<j}^N\frac{C_6}{4R_{ij}^6}  Z_i+\sum_{i<j}^N\frac{C_6}{4R_{ij}^6}  Z_j\\
&=& \sum_{i=1}^N \Big( -\frac{\delta(t)}{2}+2\sum_{j\ne i }^N\frac{C_6}{4R_{ij}^6}\Big) Z_i=- \sum_{i=1}^N\tilde{\delta}_i(t) \;Z_i
\end{eqnarray*}
So the single $Z_i$ terms can be absorbed in a redefinition of the detuning, but this latter then becomes dependent on the atom. 

For global Rabi  frequency $\Omega(t)$, detuning $\delta(t)$ and phase $\phi(t)$, we identify control terms in the Rydberg Hamiltonian as:
\begin{equation}
H(t)=H_0
+c_1(t)H_1
+c_2(t)H_2
+c_3(t)H_3
\end{equation}
with:
\begin{equation*}
H_0 =  \sum_{i<j}^N \frac{C_6}{4R_{ij}^6} ( Z_iZ_j + Z_i+Z_j),
\end{equation*}
and
\begin{equation*}
H_1 = \sum_{i=1}^NX_i,\qquad H_2 = \sum_{i=1}^N Y_i,\qquad H_3 = \sum_{i=1}^N Z_i.
\end{equation*}
This constitutes our control Hamiltonian. 

To disentangle the choice of encoding from the subsequent optimization of the physical control parameters, we construct the Krylov space using only the terms of the control Hamiltonian 
whose amplitudes are independently tunable, while excluding the geometry-dependent term $H_0$. Indeed, although $H_0$ contributes to the actual dynamics and can substantially affect the performance 
of state preparation, its coefficients are determined by the atomic geometry and therefore cannot, in general, be regarded as independently controllable parameters, particularly as the system size increases. 
Including $H_0$ in the Krylov construction would consequently make the encoding criterion dependent on a specific choice of geometry, thereby mixing two distinct optimization problems: the selection of 
a suitable encoding and the subsequent determination of a physically appropriate geometry and pulse sequence. We instead use a geometry-independent Krylov space to quantify the compatibility between 
the encoded target and the tunable control structure. The encoding is then optimized by maximizing the overlap, or equivalently minimizing the relevant principal angles, between the target subspace and 
this geometry-independent accessible subspace. This procedure is intended as a geometry-independent preselection criterion rather than as a complete characterization of the physical reachable set. 
Once an encoding has been selected, the geometry-dependent term $H_0$ is restored in the full dynamical optimization, where the geometry and pulse parameters can be varied within their 
physical constraints. In this way, the procedure separates the representation problem from the hardware/control optimization while retaining the role of geometry in the actual state-preparation dynamics.

Considering the current limitations of neutral atoms computers, the most common choice consists in starting the calculation from the ground state of the neutral-atom array: $\phi_0\rangle=|00\ldots\rangle$,
that is a state where all atoms are in their ground state. The accessible Krylov space is computed for this initial state below as an illustration for $N=2$; it is given explicitly up to 4 atoms in Appendix~\ref{app:Krylov}, and its construction can be easily automated for any $N$.

It is useful to clarify why the geometry-independent Krylov space can contain high-weight multi-qubit operators even though it is generated by collective one-body controls. 
The latter act identically on all atoms and, by themselves, do not generate entangling dynamics. Nevertheless, the initial projector $\rho_0=|0\cdots0\rangle\langle0\cdots0|$ has a 
Pauli expansion containing $Z$ strings of every weight up to $N$. Repeated commutation with the collective control Hamiltonians transforms these components into other 
permutation-symmetric operators while preserving their Pauli weight, thereby allowing multi-body operators to appear in the Krylov space. This structure is made explicit in 
Proposition~1 and in the examples for $N=2,3,4$ below. Thus, the presence of high-weight operators in $\mathcal K$ should not be interpreted as evidence that the control 
Hamiltonians themselves implement many-body interactions. Rather, it reflects the combination of the operator structure of the initial projector and the collective action of the 
controls. This distinction is important for interpreting the reachability criterion: the collective controls impose permutation symmetry, while the initial projector determines the 
range of Pauli weights represented in $\mathcal K$.

We compute the Krylov space for $N=2$ considering the initial state density matrix:
\[
\rho_0= \begin{pmatrix} 1 \\ 0 \\ 0 \\0 \end{pmatrix}(1\; 0\;0\;)^T 
=\frac{1}{4}(\mathbbm{1}+Z_0+Z_1+Z_0Z_1)
\]

The successive commutators, up to the second level are:
\begin{align*}
[X_0 + X_1, \rho_0] 
&=  -\frac{i}{2}(Y_0+Y_1 +Y_0Z_1+Z_0Y_1)\\
[Y_0 + Y_1, \rho_0]& =  \frac{i}{2}(X_0+X_1 +X_0Z_1+Z_0X_1)\\
[ X_0 + X_1,[X_0 + X_1, \rho_0] ]&= (Z_0+Z_1) -2Y_0Y_1+2Z_0Z_1)\\
[ Y_0 + Y_1,[X_0 + X_1, \rho_0] ]&= X_0Y_1+Y_0X_1\\
[ Z_0 + Z_1,[X_0 + X_1, \rho_0] ]&=-(X_0Z_1+Z_0X_1)\\
[ X_0 + X_1,[Y_0 + Y_1, \rho_0] ]&=X_0Y_1+Y_0X_1,\\
[ Y_0 + Y_1,[Y_0 + Y_1, \rho_0] ]&=(Z_0+Z_1)+2Z_0Z_1-2X_0X_1,\\
[ Z_0 + Z_1,[Y_0 + Y_1, \rho_0] ]&=-\left(Y_0Z_1+Z_0Y_1\right)
\end{align*}
Higher level commutators do not reveal any additional structure and the accessible Krylov space thus saturates after the second commutator level.
So the corresponding Krylov space is 
\begin{eqnarray*}
{\cal K}^{\rm ctrl}_m(\rho_0)&=&[(X_0 + X_1),(Y_0 + Y_1),(Z_0 + Z_1),\\
&& X_0X_1, Y_0Y_1, Z_0Z_1, \\
&& X_0Z_1+Z_0X_1,X_0Y_1+Y_0X_1,Y_0Z_1+Z_0Y_1 ]
\end{eqnarray*}
It has dimension 10 (or 9 without the identity).

A target ground state can only be reached if its projector admits an expansion on this basis, or said differently, a necessary algebraic condition for reachability
is that the density matrix of the target Hamiltonian $\rho_{targ}\in {\cal K}^{\rm ctrl}_m(\rho_0)$. Any component along antisymmetric operators such as
$X_0-X_1, Y_0-Y_1, Z_0-Z_1, X_0Y_1-Y_0X_1,\ldots$ is completely inaccessible from $\rho_0$ under the chosen controls.

Before treating $N=3$ and $N=4$ it is useful to introduce compact notation for the permutation-symmetric
operators that appear. For a composition $(n_X,n_Y,n_Z)$ of non-negative integers with $n_X+n_Y+n_Z=w\le N$, define
\begin{equation}
S^{(N)}_{n_X,n_Y,n_Z}=\!\!\sum_{\substack{A,B,C\subseteq\{1,\dots,N\}\\ A\sqcup B \sqcup C\\ |A|=n_X,|B|=n_Y,|C|=n_Z}}\!\! \prod_{i\in A}X_i \prod_{j\in B}Y_j \prod_{k\in C}Z_k,
\end{equation}
i.e.\ the sum over all ways of placing $n_X$ $X$'s, $n_Y$ $Y$'s and $n_Z$ $Z$'s on $N$ qubits (identity elsewhere). In particular $S^{(N)}_{1,0,0}=\sum_iX_i$, etc. Since $\rho_0$ and the controls $H_1,H_2,H_3$ are permutation invariant, the whole Krylov space is spanned by such symmetrized operators.

A short computation generalizing the $N=2$ case gives, for any $N$,
\begin{align}
[S_x,\rho_0]&=-\frac{i}{2^{N-1}}\sum_{k=0}^{N-1}S^{(N)}_{0,1,k},\nonumber \\
[S_y,\rho_0]&=\frac{i}{2^{N-1}}\sum_{k=0}^{N-1}S^{(N)}_{1,0,k}, \label{eq:level1}\\
[S_z,\rho_0]&=0, \nonumber
\end{align}
where $S_x=\sum_iX_i,S_y=\sum_iY_i,S_z=\sum_iZ_i$. This reproduces exactly the $N=2$ result quoted above. Two structural facts follow immediately from \eqref{eq:level1} and will organize the rest of the calculation: (i) since $[X_i,Z_i]=-2iY_i$ and $[X_i,I_i]=0$, every commutator with $S_x,S_y,S_z$ can only change the type of letter sitting on an already-active qubit, never the size of the support; hence the weight $w=n_X+n_Y+n_Z$ of a symmetrized operator is invariant along the whole commutator chain. (ii) Since $\rho_0=2^{-N}\sum_{S\subseteq\{1,\dots,N\}}Z_S$ already contains subsets $S$ of every size $0,\dots,N$, every weight sector from $1$ to $N$ is populated already at first order. The remaining task is simply to check, sector by sector, that all compositions $(n_X,n_Y,n_Z)$ of a given weight $w$ are generated.

We can formalize these observations in terms of the following proposition:
\begin{widetext}
\begin{proposition}
\label{prop:krylov}
Let $\rho_0=|0\cdots0\rangle\langle0\cdots0|$ and let the control algebra be generated by $S_x,S_y,S_z$ acting through the adjoint (commutator) representation. Then
\begin{align}
&{\cal K}_m^{\rm ctrl}(\rho_0) \nonumber\\
&=\Big\{\,S^{(N)}_{n_X,n_Y,n_Z}\ \Big|\ n_X,n_Y,n_Z\ge0,\ 1\le n_X+n_Y+n_Z\le N\,\Big\}, \nonumber
\end{align}
with
\begin{equation}
\dim{\cal K}_m^{\rm ctrl}(\rho_0)=\binom{N+3}{3}-1 \nonumber
\end{equation}
(without the identity), obtained by summing $\binom{w+2}{2}$ -- the number of compositions of $w$ into three parts --  over $w=1,\dots,N$.
\end{proposition}
\end{widetext}

\begin{proof}

(i) Weight is conserved: At each site, $X,Y,Z$ close the $\mathfrak{su}(2)$ commutation relations $[X,Y]=2iZ$ (cyclically), so $(X,Y,Z)$ transforms as a real 3-dimensional (spin-1, adjoint) vector representation under rotations generated at that site. The collective control $S_a=\sum_iA_i$ acts on a Pauli string by the same single-site rotation simultaneously at every site, so its adjoint action never turns an identity factor into a Pauli or vice versa: it can only rotate the letter already present at an occupied site among $X,Y,Z$. Hence the number of occupied sites --  the weight $w$ --  is invariant under the whole commutator chain, and the Krylov space splits into weight sectors that are each separately closed under the dynamics.

(ii) Every weight sector is seeded already at zeroth order, in $\rho_0$ itself: Writing $\rho_0=2^{-N}\sum_{S\subseteq\{1,\dots,N\}}Z_S$ with $Z_S=\prod_{k\in S}Z_k$, every subset size $|S|=0,1,\dots,N$ appears in $\rho_0$ itself, so each weight sector $w=1,\dots,N$ already contains the nonzero pure-$Z$ seed $S^{(N)}_{0,0,w}=\sum_{|S|=w}Z_S$ at order zero. Since $[X_i,Z_i]=-2iY_i\ne0$ and $[Y_i,Z_i]=2iX_i\ne0$, the first-order commutators $[S_x,\rho_0]$ and $[S_y,\rho_0]$ immediately rotate each of these seeds into a different weight-$w$ operator (one $Z$ turned into a $Y$, respectively $X$, and symmetrized), rather than reproducing the seed itself. This confirms that every weight sector, not only the top one, is already active by first order, without requiring the seed operator itself to appear in the commutator  it need not, and generically does not, since $[X_i,Z_i]\neq 0$.

(iii) Each weight sector is fully generated : For a fixed weight $w$, permutation symmetry restricts the operators to symmetric combinations of strings containing $w$ non-identity Pauli operators. Starting from the weight-$w$ string $Z_{i_1}\cdots Z_{i_w}$ identified in (ii), repeated commutators with the collective controls generate all symmetric combinations of $X$, $Y$, and $Z$ on these $w$ sites. The weight-$w$ sector is therefore fully spanned and its dimension is given by the number of non-negative triplets $(n_X, n_Y, n_Z)$ solution of the constraint $n_X+n_Y+n_Z=w$,
that is $\begin{pmatrix} w+2\\2\end{pmatrix}$.

Combining (i)-(iii): the reachable space is exactly one copy of $\mathrm{Sym}^w(\mathbb R^3)$ --  i.e.\ all compositions $(n_X,n_Y,n_Z)$ of $w$ --  for each $w=1,\dots,N$, symmetrized additionally over the $\begin{pmatrix} N\\w\end{pmatrix}$ choices of which $w$ atoms are active. Summing dimensions gives $\sum_{w=1}^N\begin{pmatrix} w+2\\2\end{pmatrix}=\begin{pmatrix} N+3\\ 3 \end{pmatrix}-1$.
\end{proof}

The construction of the accessible Krylov space can easily be automated for any $N$. Its dimension grows combinatorially as 
\begin{align} 
\mbox{dim}\;\; {\cal K}_m^{\rm ctrl}(\rho_0) -1&=\sum_{w=1}^N \begin{pmatrix} w+2\\2\end{pmatrix}=\begin{pmatrix} N+3\\3\end{pmatrix} -1 \nonumber\\
&=\frac{(N+3)(N+2)(N+1)}{6}-1.
\end{align}

The explicit construction of the Krylov space for $N=3$ and $N=4$ is given in Appendix~\ref{app:Krylov}.

\section{Application to the deuteron Hamiltonian encodings on a Rydberg machine}
\label{sec:deuteron}

\subsection{Target Hamiltonian}\label{sec:target}.

As a benchmark problem, we consider the deuteron Hamiltonian introduced by Dumitrescu et al.\cite{Dimitrescu2018}. This model has become a standard test case for quantum simulation algorithms because it provides a physically relevant nuclear many-body problem while remaining sufficiently small to be implemented on present-day quantum hardware. Moreover, its second-quantized formulation admits an exact mapping onto qubits through the Jordan-Wigner transformation, yielding sparse Hamiltonians with local one- and two-qubit interactions. These features make it an ideal benchmark for assessing the capabilities of different encoding strategies and quantum control protocols. 

For the sake of clarity, we remind here notations and definitions used in \cite{Dimitrescu2018}. The general Hamiltonian for $N$ basis states is given by

\begin{eqnarray*}
H_N &=& \sum_{n,n'=0}^{N-1} \langle n' | (T+V) | n \rangle \, a_{n'}^\dagger a_n  \label{eq:HN}.
\end{eqnarray*}
with 
\begin{equation*}
\begin{aligned}
\langle n'|T|n\rangle &= \frac{\hbar\omega}{2} \Big[ (2n + 3/2)\,\delta_n^{n'} 
- \sqrt{n(n+1/2)}\,\delta_n^{n'+1} \\
&\qquad\qquad - \sqrt{(n+1)(n+3/2)}\,\delta_n^{n'-1} \Big], \\
\langle n'|V|n\rangle &= V_0\,\delta_n^0 \delta_n^{n'} .
\end{aligned}
\end{equation*}

with $\hbar\omega=7$ MeV and $V_0=-5.68658111$ MeV \cite{Dimitrescu2018}.

This leads to:

\begin{eqnarray*}
H_N &=&\frac{\hbar\omega}{2}  \sum_{n=0}^{N-1}\Big[ (2n+\frac{3}{2})\, a_{n}^\dagger a_{n}\\
&-& \sqrt{n(n+\frac{1}{2})}\, (a_{n-1}^\dagger a_{n}+ a_{n}^\dagger a_{n-1})
\Big]+V_0\,\delta_n^0 a_{n}^\dagger a_{n}\\
&=&  \sum_{n=0}^{N-1} h_n
\end{eqnarray*}
with 

\begin{eqnarray*}
h_n &=&  \frac{\hbar\omega}{2} \Big[ (2n+\frac{3}{2})\, a_{n}^\dagger a_{n}\\
&-& \sqrt{n(n+\frac{1}{2})}\, (a_{n-1}^\dagger a_{n}+ a_{n}^\dagger a_{n-1})
\Big]+V_0\,\delta_n^0 a_{n}^\dagger a_{n}\\
\end{eqnarray*}

So we can compute $H_N$ from $H_{N-1}$ using 
$$H_N=H_{N-1}+h_{N-1}$$

\subsection{Qubit encoding}

Two distinct encoding steps are used in this work, and it is useful to distinguish them from the
outset. First, the full second-quantized Hamiltonian~(\ref{eq:HN}), acting on $N$ single-particle
orbitals, is mapped via the Jordan--Wigner transformation onto $N$ qubits, as described below. Second,
once the physical problem is restricted to its Hamming-weight-one sector (Sec.~\ref{sec:hamming1}),
the resulting $N$-dimensional single-particle Hamiltonian $H_N^{(1)}$ is re-encoded, via the standard
binary encoding, onto only $\lceil\log_2N\rceil$ qubits. This second, more economical encoding is the
one whose reachability is
analyzed in Sec.~\ref{sec:numerical}; in particular, it is why the $N=8$ Rydberg simulations of
Sec.~\ref{sec:numerical} employ only $\lceil\log_2 8\rceil=3$ atoms rather than eight. We first
recall briefly the Jordan--Wigner step.

The deuteron Hamiltonian is formulated in a basis of $N$ single-particle orbitals, each of which can be either occupied or empty. Following Ref.~\cite{Dimitrescu2018}, each orbital is encoded into one qubit using the Jordan-Wigner transformation, where the computational basis states
\[
|0\rangle \equiv \text{empty}, \qquad
|1\rangle \equiv \text{occupied}
\]
represent the occupation of the corresponding orbital. In this encoding, the fermionic number operator becomes
\[
a_i^\dagger a_i
\longrightarrow
\frac{I-Z_i}{2},
\]
while nearest-neighbor hopping terms transform as
\[
a_i^\dagger a_{i+1}+a_{i+1}^\dagger a_i
\longrightarrow
\frac{1}{2}\left(X_iX_{i+1}+Y_iY_{i+1}\right).
\]

Applying these identities to the second-quantized Hamiltonian (\ref{eq:HN}) yields the qubit Hamiltonian
\[
H_N=
\sum_{k=0}^{N-1}
 u_k (I-Z_k)
+
\sum_{k=1}^{N-1}
v_k
\left(
X_{k-1}X_k+Y_{k-1}Y_k
\right),
\]
with the coefficients 
\begin{equation*}
u_k= \frac{\hbar\omega}{4} (2k+\frac{3}{2})+\frac{V_0}{2}\delta_{k}^{0} \quad \mbox{and} \quad
v_k= - \frac{\hbar\omega}{4}\sqrt{k(k+\frac{1}{2})}.
\end{equation*}

Since we are interested in the deuteron description, the physically relevant sector to be considered is the Hamming-weight-one subspace. The corresponding reduced 
Hamiltonian is derived in the next section. 


\subsection{Hamming weight one reduction}
\label{sec:hamming1}

The deuteron Hamiltonian considered in Ref.~\cite{Dimitrescu2018} originates from the second-quantized single-particle Hamiltonian
 (\ref{eq:HN}) which acts in the one-particle sector. After the Jordan-Wigner transformation, the fermionic occupation number is mapped onto the qubit Hamming weight. 
 Since every term contains one creation and one annihilation operator, the total particle-number operator
\begin{equation*}
\hat N=\sum_{i=0}^{N-1}a_i^\dagger a_i
\end{equation*}
is conserved. Equivalently, in the qubit representation, the Hamming-weight operator
\begin{equation*}
\hat W=\sum_{i=0}^{N-1}\frac{I-Z_i}{2}
\end{equation*}
commutes with the Hamiltonian,
\begin{equation*}
[H_N,\hat W]=0.
\end{equation*}
As a consequence, the Hilbert space decomposes into invariant subspaces of fixed Hamming weight,
\begin{equation*}
\mathcal H=\bigoplus_{w=0}^{N}\mathcal H_w,
\end{equation*}
and the Hamiltonian is block diagonal with respect to this decomposition. The physical deuteron problem corresponds to a single occupied orbital and therefore lies entirely in the 
Hamming-weight-one sector,
\begin{equation*}
\mathcal H_1=\mathrm{span}\Bigl\{
|10\cdots0\rangle,\,
|010\cdots0\rangle,\,
\ldots,\,
|0\cdots01\rangle
\Bigr\},
\end{equation*}
whose dimension is $N$. Increasing the number of qubits simply enlarges the single-particle basis by introducing additional orbitals, while the particle number remains equal to one. Consequently, for any value of $N$, the ground state of the deuteron Hamiltonian is the lowest-energy eigenstate of the Hamming-weight-one block $H^{(1)}$, whereas the remaining Hamming-weight sectors correspond to states with an incorrect particle number and are therefore irrelevant for the physical problem.

Let
\begin{equation*}
|n\rangle
\equiv
|0\cdots010\cdots0\rangle,
\end{equation*}
denote the computational basis state in which only the $n$-th qubit is in the state $|1\rangle$. These $N$ states form a basis of $\mathcal H_1$.

To obtain the effective Hamiltonian acting in this subspace, we introduce the projector
\begin{equation*}
P_1=\sum_{n=0}^{N-1}|n\rangle\langle n|,
\end{equation*}
and define
\begin{equation*}
H_N^{(1)}=P_1H_NP_1.
\end{equation*}

The Jordan--Wigner Hamiltonian contains only diagonal number operators
$a_n^\dagger a_n$ and nearest-neighbour hopping terms
$a_{n-1}^\dagger a_n+a_n^\dagger a_{n-1}$.
Using
\begin{equation*}
a_n^\dagger a_n
\longrightarrow
\frac{I-Z_n}{2},
\end{equation*}
one finds immediately that
\begin{equation*}
\frac{I-Z_n}{2}|m\rangle
=
\delta_{nm}|m\rangle,
\end{equation*}
so that
\begin{equation*}
P_1
\frac{I-Z_n}{2}
P_1
=
|n\rangle\langle n|.
\label{eq:numberprojector}
\end{equation*}

Similarly, using
\begin{equation*}
X_iX_j+Y_iY_j
=
2(\sigma_i^+\sigma_j^-+\sigma_i^-\sigma_j^+),
\end{equation*}
with
\begin{equation*}
\sigma^\pm=\frac{X\pm iY}{2},
\end{equation*}
one obtains
\begin{equation*}
\sigma_i^+\sigma_j^-|j\rangle=|i\rangle,
\qquad
\sigma_i^+\sigma_j^-|k\rangle=0
\quad(k\neq j),
\end{equation*}
which implies
\begin{equation*}
P_1
\sigma_i^+\sigma_j^-
P_1
=
|i\rangle\langle j|.
\end{equation*}
Consequently,
\begin{equation*}
P_1
(X_iX_j+Y_iY_j)
P_1
=
2
\left(
|i\rangle\langle j|
+
|j\rangle\langle i|
\right).
\label{eq:XYreduction}
\end{equation*}

Substituting Eqs.~(\ref{eq:numberprojector}) and (\ref{eq:XYreduction}) into the Jordan--Wigner Hamiltonian gives
\begin{eqnarray}
H_N^{(1)}
&=&
\sum_{n=0}^{N-1}
\left[
\frac{\hbar\omega}{2}
\left(
2n+\frac32
\right)
+
V_0\delta_{n0}
\right]
|n\rangle\langle n|
\\
&
-&
\frac{\hbar\omega}{2}
\sum_{n=1}^{N-1}
\sqrt{n\left(n+\frac12\right)}
\left(
|n-1\rangle\langle n|
+
|n\rangle\langle n-1|
\right). \nonumber
\label{eq:ReducedHamiltonian}
\end{eqnarray}

Equation~(\ref{eq:ReducedHamiltonian}) is identical to the original second-quantized Hamiltonian (\ref{eq:HN}) expressed in the
single-particle basis
\(
\{|0\rangle,\ldots,|N-1\rangle\}
\).
In other words, the Hamming-weight-one reduction does not modify the Hamiltonian itself, but only its representation: the original fermionic
Fock-space description is replaced by an equivalent $N$-dimensional qubit subspace. The subsequent binary encoding of this subspace therefore
constitutes only a change of basis, allowing the physical Hilbert space to be represented with $\lceil\log_2N\rceil$ qubits instead of $N$ while
preserving exactly the spectrum and eigenstates of the physical one-particle problem.

Since the physical Hilbert space is the Hamming-weight-one sector, we henceforth restrict the Hamiltonian of Ref.~\cite{Dimitrescu2018} to this invariant subspace and denote the 
resulting block by $H^{(1)}_N$. If needed, these blocks are padded with identity blocks to obtain matrices whose dimensions are compatible with qubit encoding. 
These are the Hamiltonians whose encoding will be optimized.

\subsection{Scaling of the principal-angle optimization}
\label{sec:scaling}

The principal-angle analysis involves three distinct computational steps: the construction of the spectral algebra associated with the
target Hamiltonian, the construction of the geometry-independent Krylov space, and the optimization over the relevant unitary orbit.
These steps have different scaling properties and should therefore be distinguished when assessing the computational cost of the method.

Let
\begin{equation*}
d = 2^N
\end{equation*}
denote the Hilbert-space dimension for an $N$-qubit system. The exponential dependence of $d$ on $N$ is ultimately the main limitation
of the present implementation. However, the scaling of the individual steps depends strongly on the representation used for the Hamiltonian
and for the generated subspaces.

\paragraph{Spectral algebra.}

The target Hamiltonians considered here are naturally represented as sums of Pauli strings,
\begin{equation*}
H = \sum_{\alpha=1}^{L} h_\alpha P_\alpha,
\qquad
P_\alpha \in
\{I,X,Y,Z\}^{\otimes N},
\label{eq:pauli_decomposition_scaling}
\end{equation*}
where $L$ is the number of nonzero Pauli terms. For the Hamiltonians considered in this work, $L$ is generally much smaller than the full
number $4^N$ of Pauli strings.

The spectral algebra used in the reachability analysis is generated by successive powers of $H$,
\begin{equation*}
\mathcal{A}(H)
=
\operatorname{span}
\left\{
I,H,H^2,\ldots,H^{m-1}
\right\},
\label{eq:spectral_algebra_scaling}
\end{equation*}
where $m$ is increased until the dimension of the generated space saturates. The relevant number of powers is therefore controlled by the
minimal polynomial of $H$, and is bounded by $d$ in the absence of additional symmetries. In practice, the saturation dimension can be
substantially smaller.

The computational cost of generating these powers depends on the representation. If the operators are stored as dense $d\times d$
matrices, one matrix multiplication costs
\begin{equation*}
\mathcal{O}(d^3),
\end{equation*}
and generating $m$ powers therefore has a naive cost
\begin{equation*}
\mathcal{O}(m d^3),
\end{equation*}
with a memory requirement of $\mathcal{O}(d^2)$ per stored operator.
This rapidly becomes prohibitive because $d=2^N$.

A Pauli representation can be substantially more efficient when $L$ is small. The product of two Pauli strings is, up to a phase, another
Pauli string. Consequently, multiplication can be performed directly at the level of Pauli coefficients rather than by constructing dense
matrices. If $H$ contains $L$ Pauli strings and the number of distinct Pauli strings generated in the intermediate products is $L_k$ at the 
$k$-th power, a straightforward multiplication has a cost of order
\begin{equation*}
\mathcal{O}(L L_k)
\end{equation*}
for the multiplication by $H$. Thus the practical cost is determined not only by $N$, but also by the growth of the Pauli support under
multiplication. In the worst case, $L_k$ can approach the full Pauli operator space of dimension $4^N$, so that Pauli sparsity does not remove the exponential worst-case scaling.

For this reason, the powers of $H$ should be generated only until the corresponding vector-space dimension saturates. This avoids computing
higher powers that cannot increase the spectral algebra. 

We also emphasize again that the spectral algebra needs to be constructed only once. Under the unitary optimization,
$$
\mathcal A\!\left(H(U)\right)
=
U\,\mathcal A(H)\,U^\dagger,
$$
so its orthonormal basis can be transported by conjugation rather than recomputed at every iteration. The optimization therefore does not require repeated 
construction of the target spectral algebra; the iteration only updates its representation relative to the fixed device Krylov space.

\paragraph{Geometry-independent Krylov space.}

The Krylov space used in the present analysis is constructed from a fixed reference Hamiltonian and is therefore independent of the
atomic geometry optimized in the physical Rydberg simulation. For an initial operator $\rho_0$, we define
\begin{eqnarray*}
&&\mathcal{K}(H,\rho_0)
=
\operatorname{span}
\left\{
\rho_0,
\mathcal{L}_H(\rho_0),
\mathcal{L}_H^2(\rho_0),
\ldots
\right\},\\
&&
\mathcal{L}_H(O)=[H,O].
\label{eq:krylov_scaling}
\end{eqnarray*}

As for the spectral algebra, the construction is terminated once the dimension of the generated space ceases to increase. Since the
operator space has dimension $d^2=4^N$, the Krylov dimension is bounded by $d^2$, although it is typically much smaller because of symmetries,
conserved quantities, and the structure of the initial state and Hamiltonian.

If dense matrices are used, evaluation of one commutator requires $\mathcal{O}(d^3)$ operations. Constructing $r$ Krylov vectors therefore
has a naive cost of
\begin{equation*}
\mathcal{O}(r d^3),
\end{equation*}
in addition to the cost of orthogonalizing the generated vectors. For a dense Hilbert--Schmidt orthogonalization, the latter can become
\begin{equation*}
\mathcal{O}(r^2 d^2),
\end{equation*}
so that orthogonalization may become significant when the Krylov dimension is large.

The geometry independence of this construction is important computationally as well as conceptually. The Krylov basis is generated
only once for a given target Hamiltonian and initial state. It does not have to be recomputed during the optimization of the atomic geometry.
The subsequent geometry optimization therefore operates on a fixed subspace.

\paragraph{Principal-angle optimization over the unitary orbit.}

The final step consists of optimizing the relative orientation of the target space and the geometry-independent Krylov space under the
allowed unitary orbit. Let
\begin{equation*}
\mathcal{T}
=
\operatorname{span}\{T_1,\ldots,T_p\}
\end{equation*}
denote the target space and
\begin{equation*}
\mathcal{K}
=
\operatorname{span}\{K_1,\ldots,K_q\}
\end{equation*}
the orthonormalized Krylov space. For a unitary transformation $U\in\mathrm{U}(d)$, the optimized overlap is determined by the
principal angles between $\mathcal{T}$ and $U\mathcal{K}$.

Equivalently, defining the corresponding orthogonal projectors
$P_{\mathcal{T}}$ and $P_{\mathcal{K}}$, the optimization can be written in terms of an objective such as
\begin{equation*}
f(U)
=
\left\|
P_{\mathcal{T}} U P_{\mathcal{K}}
\right\|_{\mathrm{F}}^2.
\label{eq:principal_angle_objective_scaling}
\end{equation*}
For orthonormal bases $T$ and $K$, this is
\begin{equation*}
f(U)
=
\left\|
T^\dagger U K
\right\|_{\mathrm{F}}^2.
\end{equation*}
The singular values of $T^\dagger U K$ are the cosines of the principal angles.

The cost of one evaluation of the objective is dominated by the formation of the overlap matrix
\begin{equation*}
M(U)=T^\dagger U K.
\end{equation*}
For dense representations this scales as
\begin{equation*}
\mathcal{O}(d^2 p)
\end{equation*}
for applying the unitary to the $q$-dimensional Krylov basis, followed by
\begin{equation*}
\mathcal{O}(d p q)
\end{equation*}
for the overlap construction, plus the cost of the singular-value decomposition,
\begin{equation*}
\mathcal{O}\!\left(\min(pq^2,p^2q)\right).
\end{equation*}
The precise scaling therefore depends on the dimensions of the target and Krylov spaces, rather than on $d$ alone.

When gradient descent is performed directly on the unitary orbit, the gradient is obtained from tangent directions of the form
\begin{equation*}
U \longrightarrow e^{\epsilon A}U,
\qquad
A^\dagger=-A.
\end{equation*}
A gradient step can consequently be implemented as
\begin{equation*}
U_{n+1}
=
e^{-i\eta G_n}U_n,
\end{equation*}
where $G_n$ is the Hermitian tangent-space gradient and $\eta$ is the step size. The cost of an iteration is therefore determined by
three operations: evaluation of the objective and its gradient, application or exponentiation of the tangent-space update, and, when
required, re-orthogonalization or normalization.

For a fully dense $d\times d$ representation, these operations carry polynomial costs in $d$, typically between $\mathcal{O}(d^2)$ and
$\mathcal{O}(d^3)$ per iteration depending on how the unitary and its gradient are represented. If $n_{\mathrm{iter}}$ gradient iterations
are required, the optimization cost is correspondingly of order
\begin{equation*}
\mathcal{O}
\left(
n_{\mathrm{iter}}\,C_{\mathrm{grad}}(d,p,q)
\right),
\end{equation*}
where $C_{\mathrm{grad}}$ denotes the cost of one orbit-gradient
evaluation and update.

It is therefore useful to distinguish the scaling of the one-time construction of the geometry-independent Krylov space
from that of the iterative principal-angle optimization. Once the Krylov basis has been constructed, it is reused throughout the
optimization. The total computational cost can schematically be written as
\begin{equation*}
C_{\mathrm{total}}
=
C_{\mathrm{spectral}}
+
C_{\mathrm{Krylov}}
+
n_{\mathrm{iter}} C_{\mathrm{orbit}}.
\end{equation*}

This decomposition also clarifies the role of the geometry-independent construction. The potentially expensive generation of the Krylov space
is performed only once, whereas the optimization explores different orientations of the resulting fixed subspace. The principal-angle
calculation therefore separates the intrinsic operator-space construction from the subsequent search over the unitary orbit.

In the present work, the numerical calculations are restricted to small values of $N$. The resulting scaling nevertheless highlights an important computational challenge for direct 
implementations: the underlying Hilbert and operator spaces grow exponentially with the number of qubits. This scaling need not, however, imply that exponentially large dense 
matrices are required. In particular, the target Hamiltonians are naturally represented as sums of Pauli operators, and their structure, together with problem-specific symmetries 
and sparsity, can be exploited to reduce the computational cost. For larger systems, Pauli-based representations, symmetry reductions, tensor-network or other structured representations, and iterative Krylov techniques may therefore provide routes toward extending the approach beyond the dense-matrix regime.

\section{Numerical results}
\label{sec:numerical}

We now apply the framework of Secs.~\ref{sec:UOOpt}--\ref{sec:ControlH} to the encoded deuteron Hamiltonians introduced in Sec.~\ref{sec:hamming1}. 
The analysis proceeds in two stages. First, for each system size the encoding is analyzed and, where necessary, optimized using the principal-angle 
criterion built from the geometry-independent Krylov space $\mathcal{K}$. Second, the resulting encoding is subjected to direct dynamical optimization 
of the physical Rydberg control sequence, with the geometry-dependent interaction $H_0$ fully restored. This separation reflects the design of 
the method: the first stage provides a criterion for choosing an encoding without requiring the prior determination of an appropriate physical geometry, while the
second stage addresses the actual state-preparation problem for the chosen encoding.

Throughout, we report fidelity
\begin{equation}
f(t) = |\langle \psi_{\rm target} | \psi(t) \rangle|^2
\label{eq:fidelity_def}
\end{equation}
as the primary figure of merit, expressed as a dimensionless fraction in $[0,1]$, and we also compare the final state energy $E_{\mathrm{reached}}$ reached during the evolution 
under the control Hamiltonian, to the exact target ground state energy $E_0$.

As discussed below, these two quantities are not simple proxies for one another, and reporting both is essential to a correct interpretation of the optimization outcomes.

\subsection{Two control objectives: energy versus fidelity}
\label{sec:objectives}

The BFGS quasi-Newton algorithm~\cite{Nocedal2006} is used to optimize either of two cost functions evaluated on the terminal state $|\psi(T)\rangle$: the energy expectation value 
$\langle H(T)\rangle$, or the state infidelity $1-f(T)$. Neither is fully satisfactory in isolation.

Energy-based optimization benefits from the variational bound $\langle \psi | H | \psi \rangle \geq E_0$, which guarantees a fixed convergence target; however, 
$\langle H \rangle$ is a population-weighted average over the entire spectrum, not a measure of overlap with the target eigenvector, so a state can register a deceptively favorable or 
unfavorable energy while bearing little resemblance to the true ground state. Fidelity-based optimization removes this ambiguity by targeting the physically relevant quantity directly, 
but forgoes the variational guardrail: nothing constrains how the residual, imperfectly prepared population is distributed over the excited spectrum, so a state with substantial ground-state 
fidelity can simultaneously exhibit a much larger energy error than its fidelity alone would suggest.

A representative illustration occurs for the SB-encoded $N=3$ Hamiltonian (see Table~\ref{tab:SB34_fidobj}): fidelity-objective optimization reaches $f=0.832$, decomposing as $83.2\%$ population on the exact ground state, $\approx 0\%$ on the (unphysical, padded) first excited level, and the remaining $12.5\%$ and $4.3\%$ on two eigenstates at $E=8.56$ and $E=24.55$ in the same units 
as $E_0=-2.046$. The resulting energy, $E=0.429$ is far from $E_0$  despite the state being 
$83\%$ correct in the physically relevant sense. This is not an artifact; it is the generic behavior of $\langle H \rangle$ whenever any leakage exists onto high-lying eigenstates, and it is why 
we treat fidelity, not energy, as the primary reachability metric throughout. Energy is retained only as a secondary, cross-checking quantity, and its use as a stand-alone optimization target 
below is presented specifically to illustrate this pitfall rather than as a competing reachability criterion.

\subsection{Standard-binary encoding for $N=3$ and $N=4$}
\label{sec:SB_lack}

We first consider the two smallest nontrivial encoded systems, $N=3$ and $N=4$. 
To assess whether the SB encoding is compatible with the native Rydberg controls, we directly optimize the physical control parameters -- Rabi frequency, detuning, phase, total evolution time, 
and atomic geometry -- using BFGS, under each of the two objectives of Sec.~\ref{sec:objectives}. The results therefore test whether the target state can be prepared when the physical control problem is given maximal freedom. Because the control search is unrestricted, a systematic inability to reach high fidelity provides strong evidence of a controllability obstruction rather 
than an artifact of pulse parametrization or local optimization. The results  are summarized in Table~\ref{tab:SB34_energyobj} and \ref{tab:SB34_fidobj} .

For $N=3$, direct optimization of the Rydberg control parameters gives only a very small ground-state fidelity when the energy is minimized, whereas direct fidelity optimization 
improves the overlap to $f=0.832$ but still fails to prepare the target space with high fidelity. This limited performance indicates that the native collective control algebra is not 
sufficiently compatible with the SB representation to provide robust state preparation.

For $N=4$, the situation is considerably improved, but the SB encoding still does not yield exact state preparation.  Energy optimization gives
$E_{\mathrm{reached}}=-1.69260, \,f=0.979$, whereas direct fidelity optimization gives $f=0.981, E_{\mathrm{reached}}=-1.64548$, for a target energy 
$E_{targ}=-2.14398$. Thus, even after optimizing the dynamical control parameters, the SB encoding remains incompatible with exact preparation.

\begin{table}[htp]
\caption{SB-encoded $N=3,4$ Hamiltonians: control optimization with the energy objective, geometry free.}
\begin{center}
\resizebox{0.48\textwidth}{!}
{%
\begin{tabular}{|c|c|c|c|}\hline
 \makecell{Target \\Hamiltonian} & \makecell{Target ground\\state energy} & \makecell{Control Hamiltonian \\BFGS optimized energy} &  \makecell{Corresponding \\ Fidelity}  \\\hline\hline
$\langle H^{(1)}_{3,{\rm SB}} \rangle$ & -2.04567 &  0.0  & $10^{-8}$\\\hline
$\langle H^{(1)}_{4,{\rm SB}} \rangle$ & -2.14398 &  -1.69260  & 0.979  \\\hline
\end{tabular}
}
\end{center}
\label{tab:SB34_energyobj}
\end{table}%

\begin{table}[htp]
\caption{Same systems as Table~\ref{tab:SB34_energyobj}, control optimization with the fidelity objective, geometry free.}
\begin{center}
\resizebox{0.48\textwidth}{!}
{%
\begin{tabular}{|c|c|c|c|}\hline
 \makecell{Target \\Hamiltonian} & \makecell{Target ground\\state energy} & \makecell{Control Hamiltonian \\BFGS optimized fidelity} &  \makecell{Corresponding \\energy} \\\hline\hline
$\langle H^{(1)}_{3,{\rm SB}} \rangle$ & -2.04567 &  0.832  & 0.4288\\\hline
$\langle H^{(1)}_{4,{\rm SB}} \rangle$ & -2.14398 &  0.981  & -1.64548  \\\hline
\end{tabular}
}
\end{center}
\label{tab:SB34_fidobj}
\end{table}%

These observations motivate the use of the geometry-independent reachability criterion developed above as a diagnostic of the encoding itself.  In particular, rather than 
attempting to compensate for an incompatible encoding solely through pulse optimization, we optimize the encoding so as to align the spectral algebra of the target 
Hamiltonian with the Krylov space generated by the native collective controls.

\begin{table}[htp]
\caption{Characteristics of reachable Krylov space and target spectral algebra for standard binary encoded Hamiltonians. For 2 Rydberg atoms, we remind that dim(${\cal{K}}$)=10.}
\begin{center}
\resizebox{0.48\textwidth}{!}
{%
\begin{tabular}{|c|c|c|c|c|}\hline
 \makecell{Target \\Hamiltonian} & \makecell{dim(${\cal{A}}_{\rm targ}$)} &\makecell{${\cal A}_{\rm targ}\cap{\cal K}$} &   \makecell{Principal angles \\cosine}& \makecell{F(H)}  \\\hline\hline
$\langle H^{(1)}_{3,{\rm SB}} \rangle$ & 3  &  1  & $1.,0.835,0.534$ &-1.9829 \\\hline
$\langle H^{(1)}_{4,{\rm SB}} \rangle$ &  4 &   1 &  $1.,0.995,0.971,0.317$ & -3.024 \\\hline
\end{tabular}
}
\end{center}
\label{tab:SB34_char}
\end{table}%

Table~\ref{tab:SB34_char} makes the origin of this obstruction quantitative. For $\langle H^{(1)}_{3,{\rm SB}} \rangle$ and $\langle H^{(1)}_{4,{\rm SB}} \rangle$ , we list the dimension of the target spectral algebra $\dim(\mathcal A_{\rm targ})$, the dimension of its intersection with the reachable Krylov space, $\dim(\mathcal A_{\rm targ}\cap\mathcal K)$, the full set of principal angles cosines between the two subspaces, and the value of the subspace-overlap functional $F(H)=-\sum_i\sigma_i^2=-\sum_i\cos^2\theta_i$. Two regimes are apparent:

\begin{itemize}
\item For $N=3$, only one principal angle cosine equals~1, out of $\dim(\mathcal A_{\rm targ})=3$; the algebra intersection has dimension~1. This large misalignment ($F(H)=-1.98$, far from the saturation value $-3$) is precisely why the BFGS-optimized SB encoding fails outright.
\item For $N=4$, three of the four principal angles cosines are large but strictly less than~1 ($0.995,\,0.971$, together with a much smaller $0.317$), giving $F(H)=-3.02$ against a saturation value of $-4$. This partial, near-but-incomplete alignment is consistent with the observed near-saturated-but-imperfect fidelity of $~0.98$.
\end{itemize}

For both $N=3$ and $N=4$, this motivates the encoding optimization of Sec.~\ref{sec:opt_35}, whose explicit goal is to drive every principal angle cosine to 1 via a unitary change of representation. 

The BFGS optimization is particularly instructive because it already includes the atomic geometry as a variational
degree of freedom.  Thus, the poorer performance of the SB encoding cannot simply be attributed to the use of an unsuitable fixed
geometry.  Rather, the results show that, for these two small systems, the geometry-independent principal-angle analysis identifies an
encoding mismatch that remains relevant even after the physical geometry and pulse parameters are subsequently optimized.

\subsection{Principal-angle optimization of the encoding for $N=3$ and $N=4$}\label{sec:opt_35}

To assess the effectiveness of the proposed encoding optimization, we apply the Riemannian gradient-descent procedure of Sec.~II to the Hamming-weight-one deuteron Hamiltonians $H^{(1)}_3$ and $H^{(1)}_4$. Since all encodings related by a unitary transformation share the same spectrum, the optimization explores the unitary orbit of the Hamiltonian while leaving its physical content -- and in particular its exact ground-state energy -- unchanged.

\subsubsection{$N=3$}

The optimization is initialized from the standard permutation-symmetric (SB) encoding. For this encoding, the spectral algebra has dimension three, whereas
its intersection with the geometry-independent Krylov space has dimension only one.  The cosines of the principal angles between the two subspaces
show that only a fraction of the  three-dimensional spectral algebra lies within the controllable subspace. The single unit singular value indeed reflects the one-dimensional intersection 
between the target spectral algebra and the available geometry-independent operator space, while the remaining directions are not contained in the native control space.

Optimizing the unitary encoding removes this mismatch. The Riemannian gradient optimization converges rapidly: during the first twenty iterations the objective decreases monotonically until reaching the theoretical optimum,
\begin{equation}
F_{\rm opt}=-3=-\dim(\mathcal A(H_3)),
\end{equation}
after which the line search naturally returns a zero step size. All three principal-angle cosines reach unity at convergence, as shown in Table~\ref{tab:Opt34_BFGS}, certifying complete 
inclusion of the encoded spectral algebra within the Rydberg-reachable Krylov space, $\mathcal A(H_3^{\rm opt})\subseteq\mathcal K$.

Interestingly, the optimized Hamiltonian (Eq.~\ref{eq:H3opt}) exhibits a restored permutation symmetry between the two qubits.
This suggests that maximizing the overlap between the spectral algebra and the controllable algebra naturally drives the encoding toward representations that better exploit the intrinsic symmetry of the collective Rydberg control Hamiltonian ($S_x,S_y,S_z$ act identically on every qubit, and so naturally favor encodings symmetric under qubit exchange).
\begin{eqnarray}
H_3^{opt} &=& 7.76585\,I
-2.2954\left(
X_0 + X_1 + X_0Z_1 + Z_0X_1
\right)\nonumber\\
&+&5.71368\left(
Z_0 + Z_1
\right)-2.22862\left(
X_0X_1 + Y_0Y_1
\right)
\nonumber\\
&+&3.66151\,Z_0Z_1 
\label{eq:H3opt}
\end{eqnarray}

We then use this optimized encoding in the BFGS dynamical optimization of the Rydberg control parameters, allowing the interaction geometry to vary.  
The resulting protocol reaches the exact ground-state energy $E_0=-2.04567$ with $100\%$ fidelity (see Table~\ref{tab:Opt34_BFGS}). Remarkably, fixing the interaction 
geometry and repeating the dynamical optimization yields the same exact ground-state energy and 100\% fidelity (see Sec.~\ref{subsec:fixed_geometry_SB}). 
This demonstrates that the success of the optimized encoding does not rely on optimizing the geometry and is therefore robust to the choice of interaction geometry.
optimized jointly with the control parameters.

\subsubsection{$N=4$}

The same procedure applied to $H_4^{(1)}$, initialized from the SB encoding with $\dim(\mathcal A(H_4))=4$ and $F_0=-3.02$ (see Table~\ref{tab:SB34_char}).
Again, the spectral algebra has dimension four, while its intersection with the geometry-independent Krylov space is one-dimensional.

After optimization of the encoding, all four principal-angle cosines reach unity, as shown in Table~\ref{tab:Opt34_BFGS}, which establishes
$ A(H^{\mathrm{opt}}_3)\subseteq K$.

The resulting optimized Hamiltonian $H_4^{\rm opt}$ is given in Pauli decomposition below. 
\begin{eqnarray*}
H^{opt}_4&=&14.3284\,\mathbbm{1}
-4.7639\,(X_0+X_1)
+8.14177\,(Z_0+Z_1)
\\
&-&1.98892\,X_0X_1
-2.4656\,(X_0Z_1+Z_0X_1)
\\
&-&2.47555\,Y_0Y_1
+1.04113\,Z_0Z_1 
\end{eqnarray*}

Re-optimizing the Rydberg control sequence for $H_4^{\rm opt}$ yields a ground-state energy of $-2.143939$ against the exact value $-2.143982$, a relative error below $2\times10^{-5}$, corresponding to $100\%$ overlap fidelity within our numerical resolution -- compared to the $98\%$ obtained with the unoptimized SB encoding of Sec.~\ref{sec:SB_lack}. More interestingly,
as for $N=3$, re-running the optimization with fixed interaction geometry also allows to reach unit fidelity.

Together, the $N=3$ and $N=4$ results confirm that (i) the principal-angle objective is a reliable predictor of control failure, and (ii) when misalignment is detected, the Riemannian optimization systematically repairs it, recovering exact reachability without altering the physical spectrum.

\begin{table}[htp]
\caption{Optimized ground state energies and corresponding overlap fidelity reached by the Rydberg atoms evolution using optimized Hamiltonians.}
\begin{center}
\resizebox{0.48\textwidth}{!}
{%
\begin{tabular}{|c|c|c|c|c|c|}\hline
 \makecell{Target \\Hamiltonian} & \makecell{Target ground\\state energy} & \makecell{Control Hamiltonian \\optimized energy} &  \makecell{Fidelity} & \makecell{Principal \\angles cosine} & F(H) \\\hline\hline
$\langle H^{(1)}_{3,{\rm opt}} \rangle$ & -2.04567 &   - 2.04567  & 100\%  & $1.,1.,1.$ &-3 \\\hline
$\langle H^{(1)}_{4,{\rm opt}} \rangle$ & -2.143982 &  -2.143939  & 100\%& $1.,1.,1.,1.$ & -4 \\\hline
\end{tabular}
}
\end{center}
\label{tab:Opt34_BFGS}
\end{table}%

These results provide a direct numerical validation of the proposed encoding criterion for the smallest nontrivial cases.

These two examples provide a direct test of the central idea of the procedure.  The principal-angle optimization is performed without
knowledge of, or optimization over, the atomic geometry.  Nevertheless, the encoding selected by this geometry-independent criterion performs
substantially better in the subsequent physical control optimization, where the geometry is restored and optimized.  Thus, the role of the
principal-angle procedure is not to replace the dynamical optimization or to determine the optimal geometry.  Instead, it provides a
geometry-independent criterion for identifying a representation of the target Hamiltonian that is structurally better matched to the
independently tunable control operators.

\subsection{Standard-binary encoding for $N=5,\ldots,8$}
\label{subsec:largeN_SB}

\subsubsection{Algebraic reachability and state-level reachability}
\label{subsubsec:LargeN}

The behavior for larger systems reveals an important distinction between algebraic reachability and state-level reachability.  For the SB encoding,
the geometry-independent principal-angle criterion does not reach full containment for $N=5,\ldots,8$.  The principal-angle cosines obtained for
the unoptimized SB encoding, together with their corresponding objective $F(H)$, are summarized in Table~\ref{tab:SB_char}.

\begin{table}[htp]
\caption{Characteristics of reachable Krylov space and target spectral algebra for standard binary encoded Hamiltonians. For 3 Rydberg atoms, dim(${\cal{K}}$)=20.}
\begin{center}
\resizebox{0.48\textwidth}{!}
{%
\begin{tabular}{|c|c|c|c|c|}\hline
 \makecell{Target \\Hamiltonian} & \makecell{dim(${\cal{A}}$)} &\makecell{${\cal A}_{\rm targ}\cap{\cal K}$} &  \makecell{Principal angles\\ cosine}& \makecell{F(H)}  \\\hline\hline
$\langle H^{(1)}_{5,{\rm SB}} \rangle$ & 5 &  1  &$1.,0.722 ,0.515 ,0.408,0.213$ &  -1.99967\\\hline
$\langle H^{(1)}_{6,{\rm SB}} \rangle$ & 6 & 1   &  $1.,0.854,0.528,0.455,0.369,0.274$ & -2.42701 \\\hline
$\langle H^{(1)}_{7,{\rm SB}} \rangle$ & 7 & 1   &  $1.,0.840,0.601,0.585,0.422,0.338,0.242$ & -2.76059 \\\hline
$\langle H^{(1)}_{8,{\rm SB}} \rangle$ & 8 & 1   &  $1.,0.947,0.922,0.643,0.582,0.423,0.277,0.123$ & -3.76885 \\\hline
\end{tabular}
}
\end{center}
\label{tab:SB_char}
\end{table}%

For all four system sizes, the intersection between the spectral algebra and the geometry-independent Krylov space remains one-dimensional, whereas
the spectral-algebra dimensions are respectively $5$, $6$, $7$, and $8$.

Nevertheless, when the geometry of the Rydberg array is included as a free optimization variable, the SB encoding can reach the desired ground state
with very high fidelity.  The results are reported in Table~\ref{tab:SB_BFGS_geom}.
For all four system sizes, the optimized dynamical protocols reproduce the corresponding ground-state energies within the reported numerical
precision and achieve fidelities larger than $0.999$.

\begin{table}[htp]
\caption{Optimized ground state energies and corresponding overlap fidelity reached by the Rydberg atoms evolution using Standard Binary encoding. In these results, the physical control -- including atoms geometry -- 
are left entirely free and are all optimized by the BFGS procedure.}
\begin{center}
\resizebox{0.48\textwidth}{!}
{%
\begin{tabular}{|c|c|c|c|}\hline
 \makecell{Target \\Hamiltonian} & \makecell{Target ground\\state energy} & \makecell{Control Hamiltonian \\BFGS optimized energy} &  \makecell{Fidelity}  \\\hline\hline
$\langle H^{(1)}_{5,{\rm SB}} \rangle$ & -2.183592&  -2.183592&$>0.999$\\\hline
$\langle H^{(1)}_{6,{\rm SB}} \rangle$ &- 2.201569&-2.201569 & $>0.999$ \\\hline
$\langle H^{(1)}_{7,{\rm SB}} \rangle$ &- 2.21042& - 2.21042  & $>0.999$\\\hline
$\langle H^{(1)}_{8,{\rm SB}} \rangle$ &-2.215039& -2.215039 & $>0.999$\\\hline
\end{tabular}
}
\end{center}
\label{tab:SB_BFGS_geom}
\end{table}%

This result is significant because it demonstrates that the failure of the geometry-independent algebraic containment condition does not imply failure
of ground-state preparation.  The condition $A(H)\subseteq K$ is a sufficient, geometry-independent compatibility criterion at the operator-algebra 
level, but it is stronger than what is required to prepare a single eigenstate when the geometry-dependent interaction Hamiltonian is restored and 
optimized. Thus the distinction between algebra-level and state-level reachability -- together with the deliberate exclusion of $H_0$ from
$\mathcal K$ discussed in Sec.~III. --  is essential. 
The geometry-independent Krylov space $\mathcal K$ indeed characterizes only the operator directions reachable under the independently tunable global controls $S_x,S_y,S_z$; it does
not include the additional directions opened up once the interaction term $H_0(\mathbf r)$ is restored and the geometry $\mathbf r$ is allowed to vary. When geometry is optimized jointly with the pulse sequence, the effective reachable space explored by the dynamics is larger than $\mathcal K$, and evidently large enough to contain the specific ground-state direction of interest for $N=5,\ldots,8$, even though it is not large enough to contain the entire spectral algebra $\mathcal A(H_N^{(1)})$ in the geometry-independent sense measured by the diagnostic. 
Complete algebra containment in $\mathcal K$ is thus confirmed to be sufficient, but -- as soon as geometry is a free parameter -- not necessary, for high-fidelity ground-state preparation.
In other words, the geometry can provide additional dynamical resources that are not captured by the geometry-independent reachability 
criterion.

\subsubsection{Sensitivity of the SB encoding to the interaction geometry}
\label{subsec:fixed_geometry_SB}

The results above for the principal angle optimized Hamiltonians $N=3$ and $N=4$ were obtained with a fixed interaction geometry, providing a stringent test of the optimized
encodings.  We now investigate systematically the dependence of the standard-binary encoding on the Rydberg geometry for the larger target
spaces $N=5,\ldots,8$.

For each system size, we consider the same five representative geometries,
\begin{align}
    \mathbf r_A &= (0,0),(12,0),(24,0), \nonumber\\
    \mathbf r_B &= (0,0),(12,0),(32,0), \nonumber\\
    \mathbf r_C &= (0,0),(9,9),(18,0), \nonumber\\
    \mathbf r_D &= (0,0),(9,4),(18,0), \nonumber\\
    \mathbf r_E &= (0,0),(18,9),(18,0).
\end{align}
For each fixed geometry, the encoding is kept unchanged and the BFGS optimization is restricted to the dynamical control parameters
$\Omega$, $\delta$, $\phi$, and the total evolution time $T$.  Thus, no geometrical degree of freedom is available to the optimizer.

We perform the control optimization using both the final energy and the ground-state fidelity as objective functions.  

For the unoptimized SB encoding, the fidelity optimization results are shown in Table~\ref{tab:SB_fixed_geom3} and Table~\ref{tab:SB_fixed_geom4}.  The strong dependence on
geometry is already apparent at the fidelity level: depending on the geometry, the optimized fidelity ranges from 0.60 to 0.99. It is even more striking 
at the level of the energy.  For example, for $N=5$ the optimized energy ranges from positive values,
$E_{\rm reached}=5.92$ for $\mathbf r_A$, $\mathbf r_C$, and $\mathbf r_D$, to $E_{\rm reached}=-1.3860$ for $\mathbf r_E$, while the exact ground-state
energy is $E_0=-2.18359$.  Similarly, for $N=7$, the optimized energy varies from $E_{\rm reached}=6.2784$ at $\mathbf r_A$, $\mathbf r_C$, and
$\mathbf r_D$ to $E_{\rm reached}=-2.05439$ at $\mathbf r_E$, compared with $E_0=-2.21042$.  The effect is less dramatic for $N=8$, for which geometries
$\mathbf r_B$ and $\mathbf r_E$ yield energies very close to the exact ground-state value, whereas the other geometries remain substantially higher.

\begin{table}[htp]
\caption{BFGS optimized fidelity of the (unoptimized) SB-encoded ground state reached by BFGS control optimization at
five fixed atomic geometries, for $N=5,6,7,8$. Geometry and control parameters other than $\Omega,\delta,\phi,T$ are excluded from the search.}
\begin{center}
\resizebox{0.36\textwidth}{!}{%
\begin{tabular}{|c|c|c|c|c|c|}\hline
$N$ & $\mathbf r_A$ & $\mathbf r_B$ & $\mathbf r_C$ & $\mathbf r_D$ & $\mathbf r_E$ \\\hline\hline
5 & $0.603$ & $0.841$ & $0.603$ & $0.603$ & $0.924$ \\\hline
6 & $0.933$   & $0.999$  & $0.933$    & $0.933$   & $0.931$ \\\hline
7 & $0.679$ & $0.996$    & $0.679$  & $0.679$ & $0.993$ \\\hline
8 & $0.951$   & $0.999$& $0.951$    & $0.951$   & $>0.999$ \\\hline
\end{tabular}
}
\end{center}
\label{tab:SB_fixed_geom3}
\end{table}%

\begin{table}[htp]
\caption{Corresponding energy of the (unoptimized) SB-encoded ground state reached by BFGS control optimization at
five fixed atomic geometries, for $N=5,6,7,8$. Geometry and control parameters other than $\Omega,\delta,\phi,T$ are excluded from the search.}
\begin{center}
\resizebox{0.46\textwidth}{!}{%
\begin{tabular}{|c|c|c|c|c|c|}\hline
$N$ & $\mathbf r_A$ & $\mathbf r_B$ & $\mathbf r_C$ & $\mathbf r_D$ & $\mathbf r_E$ \\\hline\hline
5 & $5.92$ & $1.74$ & $5.92$ & $5.92$ & $-1.3860$ \\\hline
6 & $-0.90686$   & $-2.20157$  & $-0.90686$    & $-0.90686$   & $-1.67285$ \\\hline
7 & $6.2784$ & $-2.13088$    & $6.2784$  & $6.2784$ & $-2.05439$ \\\hline
8 & $-0.64587 $   & $-2.21356$& $-0.645744$    & $-0.6457688$   & $-2.21504$ \\\hline
\end{tabular}}
\end{center}
\label{tab:SB_fixed_geom4}
\end{table}%

The energy-minimisation leads to the results summarized in Table~\ref{tab:SB_fixed_geom1} and  Table~\ref{tab:SB_fixed_geom2}.  The geometry dependence is even
more directly apparent in these results.  For $N=5$, the optimized SB encoding reaches fidelities below $10^{-9}$ at $\mathbf r_A$, $\mathbf r_C$, and $\mathbf r_D$, and below $10^{-10}$ at $\mathbf r_B$, while reaching $F=0.924$ at $\mathbf r_E$.  For $N=6$, the fidelities are $F\simeq0.933$ at four of the five geometries but fall below $10^{-6}$ at
$\mathbf r_B$.  For $N=7$, the same geometries $\mathbf r_A$, $\mathbf r_C$, and $\mathbf r_D$ again lead to very poor fidelities,
whereas $\mathbf r_B$ and $\mathbf r_E$ yield $F=0.996$ and $0.993$, respectively.  At $N=8$, by contrast, the SB encoding remains comparatively
robust, with $F\geq0.928$ for all five geometries and fidelities above $0.999$ for $\mathbf r_B$ and $\mathbf r_E$.

\begin{table}[htp]
\caption{Fidelity of the (unoptimized) SB-encoded ground state reached by BFGS control optimization at
five fixed atomic geometries, for $N=5,6,7,8$. Geometry and control parameters other than $\Omega,\delta,\phi,T$ are excluded from the search.}
\begin{center}
\resizebox{0.38\textwidth}{!}{%
\begin{tabular}{|c|c|c|c|c|c|}\hline
$N$ & $\mathbf r_A$ & $\mathbf r_B$ & $\mathbf r_C$ & $\mathbf r_D$ & $\mathbf r_E$ \\\hline\hline
5 & $<10^{-9}$ & $<10^{-10}$ & $<10^{-8}$ & $<10^{-7}$ & $0.889$ \\\hline
6 & $0.919$   & $<10^{-6}$  & $0.919$    & $0.919$   & $0.620$ \\\hline
7 & $<10^{-9}$ & $0.652$    & $<10^{-8}$  & $<10^{-7}$ & $0.806$ \\\hline
8 & $0.928$   & $>0.999$& $0.928$    & $0.928$   & $0.999$ \\\hline
\end{tabular}}
\end{center}
\label{tab:SB_fixed_geom1}
\end{table}%

\begin{table}[htp]
\caption{BFGS optimized energy of the (unoptimized) SB-encoded ground state reached by BFGS control optimization at
five fixed atomic geometries, for $N=5,6,7,8$. Geometry and control parameters other than $\Omega,\delta,\phi,T$ are excluded from the search.}
\begin{center}
\resizebox{0.48\textwidth}{!}{%
\begin{tabular}{|c|c|c|c|c|c|}\hline
$N$ & $\mathbf r_A$ & $\mathbf r_B$ & $\mathbf r_C$ & $\mathbf r_D$ & $\mathbf r_E$ \\\hline\hline
5 & $0.0$ & $0.0$ & $0.0$ & $0.0$ & $-1.3944$ \\\hline
6 & $-1.32125$   & $0.0$  & $-1.32125$    & $-1.32125$   & $-1.35360$ \\\hline
7 & $0.0$ & $-1.09564$    & $0.0$  & $0.0$ & $-1.24073$ \\\hline
8 & $-1.37659$   & $-2.21504$& $-1.37659$    & $-1.37659$   & $-2.21419$ \\\hline
\end{tabular}
}
\end{center}
\label{tab:SB_fixed_geom2}
\end{table}%

The energy and fidelity optimizations therefore lead to the same qualitative conclusion: the performance of the unoptimized SB representation can depend
strongly on the interaction geometry when the latter is held fixed.  

This fixed-geometry behavior is consistent with the geometry-independent principal-angle analysis, but also illustrates its limitations.  For
$N=5,\ldots,8$, the geometry-independent Krylov-space diagnostic identifies persistent algebraic misalignment of the SB encoding, with only one
principal angle cosine equal to unity.  It therefore provides no geometry-independent guarantee of successful state preparation.  At fixed geometry,
the actual performance depends additionally on the alignment of the geometry-dependent interaction Hamiltonian $H_0(\mathbf r)$ with the target
ground state.  This dependence is not included in the construction of the geometry-independent Krylov space and is therefore not expected to be
captured by the principal-angle diagnostic alone.

An instructive example is $N=8$.  Despite having no stronger geometry-independent algebraic containment than the smaller systems, the
SB encoding never falls below $92.8\%$ fidelity over the five geometries considered here.  The particularly good performance at $\mathbf r_B$ and
$\mathbf r_E$ is therefore evidence of a favorable geometry-dependent alignment rather than of improved algebraic compatibility in the
geometry-independent criterion.  Conversely, the poor performance at selected geometries for $N=5$--$7$ demonstrates that high fidelity obtained
in the fully unconstrained optimization of Sec.~\ref{subsubsec:LargeN} does not imply robustness with respect to the physical geometry.

\subsection{Geometry-independent encoding optimization at fixed geometry}
\label{subsec:optimized_encoding_fixed_geometry}

We next ask whether the geometry sensitivity identified above can be reduced by optimizing the encoding itself, without using any information about the
physical geometry.  For each $N=5,\ldots,8$, a single encoding is obtained from the principal-angle optimization of Sec.~\ref{sec:UOOpt}.
We emphasize that this principal-angle optimization depends only on the target Hamiltonian and the geometry-independent Krylov space; in particular, 
the atomic positions are not included as optimization variables. Once obtained, the optimized encoding is thus kept fixed and reused for all five geometries.

Table~\ref{tab:Opt5678_BFGS} summarizes the principal-angle cosines and the corresponding overlap objective after optimization. The table should be compared directly with the standard-binary results of Table~\ref{tab:SB_char}. The improvement is systematic: the optimization increases the number of directions of the target spectral algebra that are exactly aligned with the 
geometry-independent Krylov space. Specifically, the number of unit principal-angle cosines increases from 1 for the SB encoding to 3, 4, 5, and 6 for $N=5,6,7,8$, respectively. Thus, 
although complete algebra containment is not reached for these system sizes, the optimized encoding places a progressively larger fraction of the target spectral algebra inside the operator 
space generated by the independently tunable collective controls. The residual non-unit singular values quantify the directions that remain outside this geometry-independent space. 
This distinction is important: the optimization is not selecting a favorable geometry, nor is it claiming complete state-level controllability; it is improving the algebraic compatibility of the 
representation with the control resources that are common to all geometries.

\begin{table}[htp]
\caption{ Principal angle cosines and objective functions for $N=5$ to $N=8$. The optimization is performed over the unitary orbit of the target Hamiltonian using the 
geometry-independent Krylov space $\mathcal K$.}
\begin{center}
\resizebox{0.48\textwidth}{!}
{%
\begin{tabular}{|c|c|c||}\hline
 \makecell{Target \\Hamiltonian} &  \makecell{Principal angles \\ cosine} & F(H) \\\hline\hline
$\langle H^{(1)}_{5,{\rm opt}} \rangle$ &  $1.,1.,1.,0.886,0.630 $&  -4.18241 \\\hline
$\langle H^{(1)}_{6,{\rm opt}} \rangle$ &$1.,1.,1.,1.,0.930,0.091$ &-4.87249 \\\hline
$\langle H^{(1)}_{7,{\rm opt}} \rangle$ &$1.,1.,1.,1.,1.,0.622791,4.10^{-7}$ &-5.38787  \\\hline
$\langle H^{(1)}_{8,{\rm opt}} \rangle$ & $1.,1.,1.,1.,1.,1.,7.10^{-7},1.10^{-8}$ &-6  \\\hline
\end{tabular}
}
\end{center}
\label{tab:Opt5678_BFGS}
\end{table}%
For each geometry, we again perform the BFGS optimization of the dynamical control parameters $\Omega$, $\delta$, $\phi$, and $T$, with the geometry
held fixed. The optimized encoding is not changed between these runs.

The resulting ground-state optimized fidelities are reported in Table~\ref{tab:opt_fixed_geom}.  For $N=5$, the same
optimized encoding gives $F>0.9999$ at four of the five geometries and $F=0.9753$ at $\mathbf r_E$.  For $N=6$, the fidelity remains above
$0.996$ at four geometries and is $0.9409$ at $\mathbf r_E$.  For $N=7$, four geometries again give $F>0.9999$, while the fidelity at $\mathbf r_E$
is $0.8596$.  Finally, for $N=8$, all five geometries yield fidelities above $0.9997$, with four of them above $0.9999$. 

\begin{table}[htp]
\caption{Fidelity of the principal-angle-optimized encoding at the same five fixed atomic geometries given in the text. A single optimized Hamiltonian per $N$ is reused, unchanged, across all
five geometries.}
\begin{center}
\resizebox{0.48\textwidth}{!}{%
\begin{tabular}{|c|c|c|c|c|c|}\hline
$N$ & $\mathbf r_A$ & $\mathbf r_B$ & $\mathbf r_C$ & $\mathbf r_D$ & $\mathbf r_E$ \\\hline\hline
5 & $>0.9999$ & $>0.9999$ & $>0.9999$   & $>0.9999$   & $0.9753$ \\\hline
6 & $>0.9999$ & $>0.9963$   & $>0.9999$   & $>0.9999$   & $0.9409$ \\\hline
7 & $>0.9999$ & $>0.9999$ & $>0.9999$   & $>0.9999$   & $0.8596$ \\\hline
8 & $>0.9999$ & $>0.9997$   & $>0.9999$   & $>0.9999$   & $>0.9999$ \\\hline
\end{tabular}}
\end{center}
\label{tab:opt_fixed_geom}
\end{table}%

The contrast with the unoptimized encoding is substantial. Across the fixed-geometry benchmark, the principal-angle-optimized encoding consistently exhibits high
fidelity, demonstrating that the improvement identified by the geometry-independent algebraic criterion translates into robust state-level performance after 
the physical geometry is restored.

In the most extreme cases, namely $N=5$ and $N=7$ at $\mathbf r_A$, $\mathbf r_C$, and
$\mathbf r_D$, the fidelity increases from below $10^{-7}$ to above $0.9999$ without changing the physical geometry, the set of available 
control Hamiltonians, or the target spectrum.  The improvement therefore cannot be attributed to the selection of a more favorable geometry. It is a consequence of 
changing only the representation of the same target Hamiltonian.

The fact that a single optimized encoding is reused across all five geometries is particularly significant.  The encoding is determined before
the geometry is specified and is not re-optimized for any of the subsequent control problems.  The simultaneous increase in the number of exactly aligned algebraic 
directions and in the worst-case state-preparation fidelity therefore demonstrates that the principal-angle procedure provides a genuine geometry-independent improvement 
in the representation, rather than merely identifying a geometry for which the original SB encoding happens to perform well.

The residual geometry dependence is also informative.  The geometry $\mathbf r_E$ is comparatively more difficult for the optimized encodings
at $N=5$, $6$, and $7$, for which the fidelities are $0.9753$, $0.9409$, and $0.8596$, respectively.  This residual dependence is expected because the
principal-angle optimization does not include the geometry-dependent interaction $H_0(\mathbf r)$ in its objective.  It optimizes the alignment
of the target spectral algebra with the geometry-independent controllable operator space, but does not enforce simultaneous compatibility with every
possible realization of $H_0(\mathbf r)$.  The remaining spread in fidelity therefore provides a direct measure of the part of the physical control problem that is 
geometry specific and intentionally left to the subsequent dynamical optimization.

Taken together, these results show that principal-angle optimization acts as a geometry-independent preprocessing step for analog quantum simulation.
It does not determine the optimal physical geometry and does not replace the subsequent dynamical control optimization.  Instead, it provides a
representation in which the target Hamiltonian is substantially better matched to the native collective controls before any device-specific
geometry is selected.  This separation is particularly valuable when the geometry is fixed by the hardware or when simultaneous optimization of the
encoding, geometry, and pulse parameters becomes computationally expensive.

The numerical results establish three complementary points.

First, optimizing the encoding can transform an operator representation that is poorly aligned with the available collective controls into one for
which the geometry-independent spectral algebra is fully contained in the corresponding Krylov space, as demonstrated explicitly for $N=3$ and
$N=4$.  In these cases, the algebraic optimization is followed by near-perfect dynamical state preparation.

Second, the larger-$N$ results show that the algebraic containment criterion should not be interpreted as a necessary condition for preparing a single
ground state.  For $N=5,\ldots,8$, the SB encoding does not satisfy the geometry-independent containment condition, yet joint optimization of the
Rydberg geometry and dynamical controls still produces fidelities exceeding $0.999$.  This is possible because the geometry-dependent interaction
Hamiltonian provides additional structure that is absent from the geometry-independent Krylov analysis.

Third, the fixed-geometry calculations demonstrate the practical advantage of optimizing the encoding independently of the device geometry.  For the
unoptimized SB encoding, the achievable fidelity varies strongly among different geometries.  After geometry-independent encoding optimization,
the same encoding yields high fidelity over all five tested geometries, with minimum fidelities of $0.9753$, $0.9409$, $0.8596$, and approximately
$0.9997$ for $N=5$, $6$, $7$, and $8$, respectively.

Taken together, these results support the interpretation of the principal-angle construction as a geometry-independent encoding-design
criterion rather than as a complete characterization of state-level controllability.  It identifies operator-space incompatibilities before
any device-specific geometry or pulse optimization is performed, and its optimization can substantially reduce the subsequent dependence on the
physical interaction geometry.  The criterion therefore provides a systematic preprocessing step for analog quantum simulation: the encoding
can be optimized at the algebraic level first, after which the remaining device-specific optimization is restricted to the physical geometry and
control parameters.

\subsection{Explicit Rydberg control protocols}

The analysis above provides a necessary reachability criterion, but it does not by itself determine whether a particular
physical control protocol can prepare the target state. We therefore complement it with direct simulations of the Rydberg dynamics using the
optimized control parameters and illustrate the results for the $N=3$ case. 

Under the SB encoding, the system is unable to reach the target ground state, as illustrated in Fig.~\ref{fig:H3Simu}, no matter which criterion (energy-minimization (l.h.s.) or fidelity-maximization (r.h.s.)) is used of optimization, with only one of three principal angles equal to 1, the target ground state of $H_3^{(1)}$ has essentially no overlap with directions reachable under the native 
Rydberg controls, so no choice of pulse sequence, however extensively optimized, can prepare it. The low fidelity is not an optimization failure but a controllability failure.

To make the algebra-reachability obstruction identified in Sec. V A more concrete, we contrast the explicit Rydberg dynamics obtained for the two encodings of $H_3^{(1)}$ side by side in 
Fig.~\ref{fig:H3Simu_opt}. For standard binary encoding, we seek for the best fidelity by including in the optimization all possible control parameters, including atomic geometry, whereas
for optimal Hamiltonian $H_{3, \rm opt}^{(1)}$, we keep the geometry fixed, making the BFGS optimization harder in principle. 
Under the principal-angle-optimized encoding, the system evolves toward  the correct target ground state superposition (represented by dots at the final time evolution in Fig~\ref{fig:H3Simu_opt}), 
the energy reaches the exact ground-state value, and the run terminates at 99.999\% fidelity (Table~\ref{tab:Opt34_BFGS}). The results of these noiseless simulations thus demonstrate that the 
target ground state can be explicitly prepared.

These results provide a direct dynamical confirmation that the principal-angle criterion of Sec. II identifies not just a mathematical obstruction but an operationally decisive one.

\begin{figure}[htbp]
\begin{center}
\includegraphics[height=3.5cm,width=4.2cm]{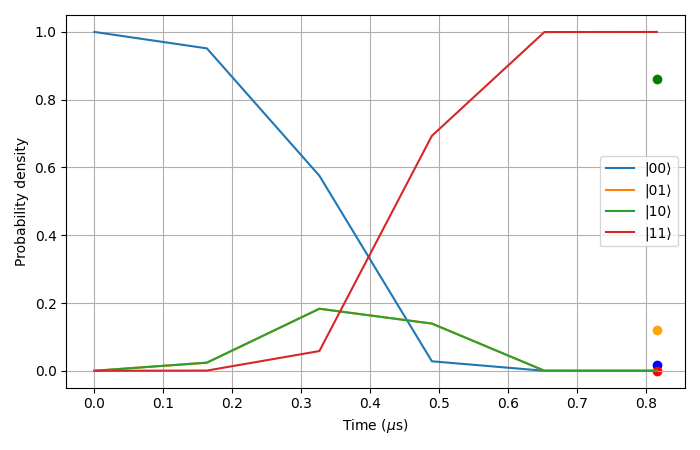}  
\includegraphics[height=3.5cm,width=4.2cm]{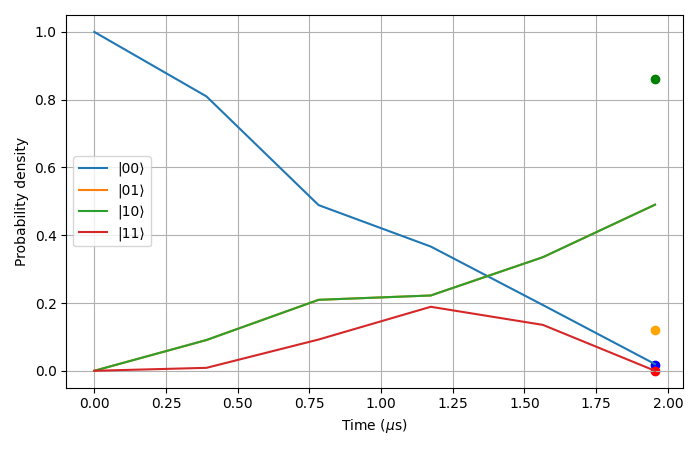}  
\includegraphics[height=3.5cm,width=4.2cm]{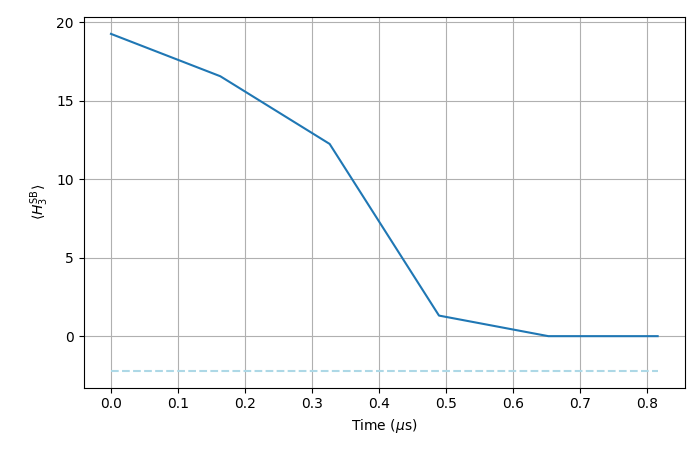}  
\includegraphics[height=3.5cm,width=4.2cm]{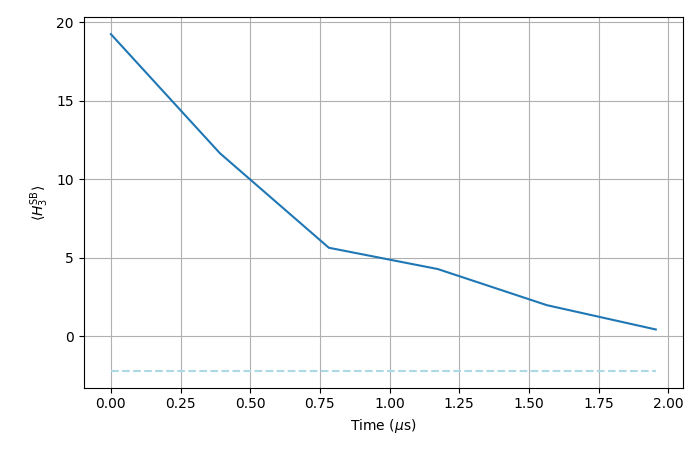}  
\includegraphics[height=3.5cm,width=4.2cm]{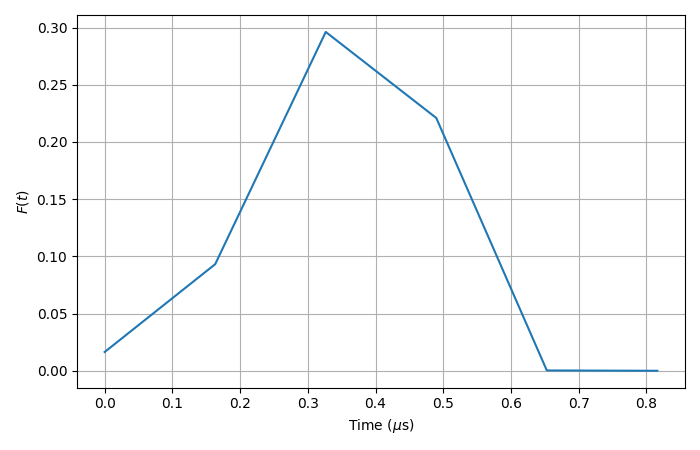}  
\includegraphics[height=3.5cm,width=4.2cm]{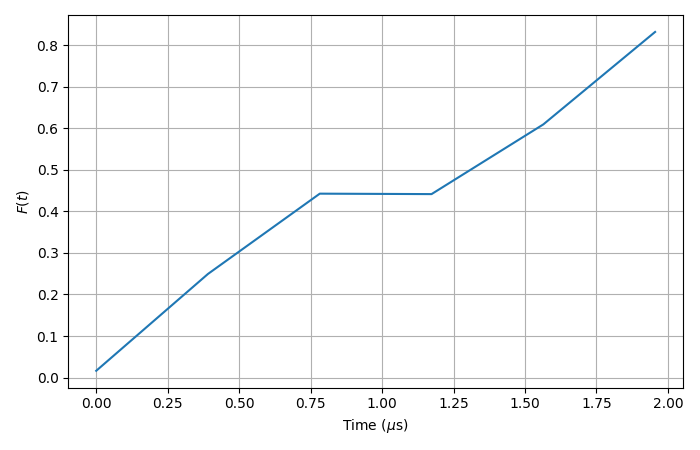}  
\caption{{\it{ Evolution of three Rydberg atoms with optimized parameters to reach the ground state of the deuteron Hamiltonian $H^{(1)}_{3,{\rm SB}}$. 
The top panel displays the probability density
evolution, and the bottom one the corresponding energy. The dashed line corresponds to the exact deuteron ground state energy. The Rydberg Hamiltonian for 2 atoms being symmetric in the exchange of the atoms, the states $|01\rangle$ and $|10\rangle$ are degenerated. The exact deuteron ground state is represented by colored dots. 
}}}\label{fig:H3Simu}
\end{center}
\end{figure}

\begin{figure}[htbp]
\begin{center}
\includegraphics[height=3.5cm,width=4.2cm]{ProbabilityH3SB_FidelityGeom.png}  
\includegraphics[height=3.5cm,width=4.2cm]{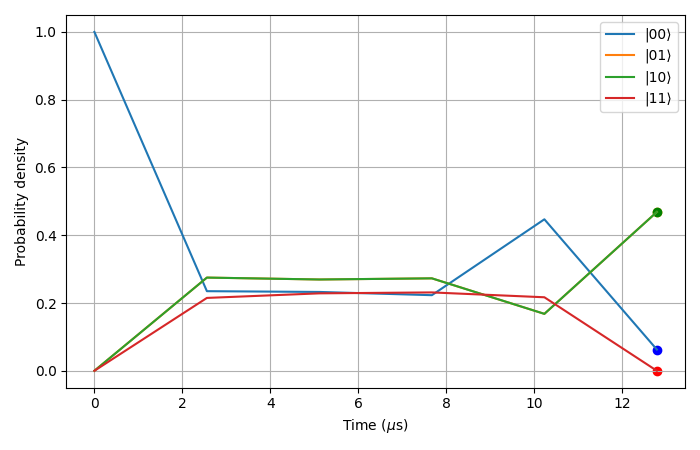}  
\includegraphics[height=3.5cm,width=4.2cm]{EnergyH3SB_FidelityGeom.png}  
\includegraphics[height=3.5cm,width=4.2cm]{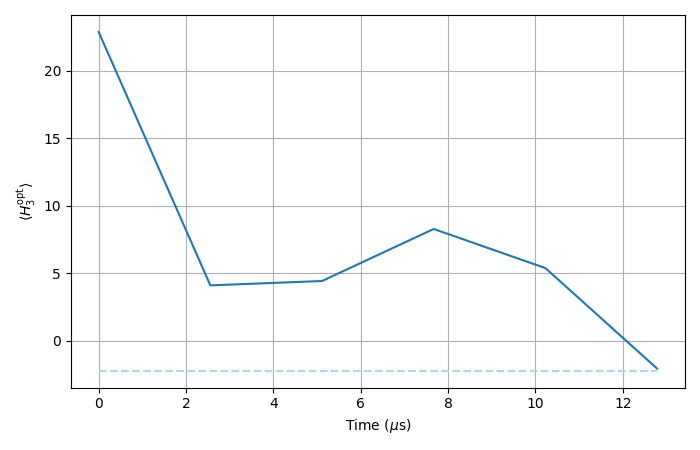}  
\includegraphics[height=3.5cm,width=4.2cm]{FidelityH3SB_FidelityGeom.png}  
\includegraphics[height=3.5cm,width=4.2cm]{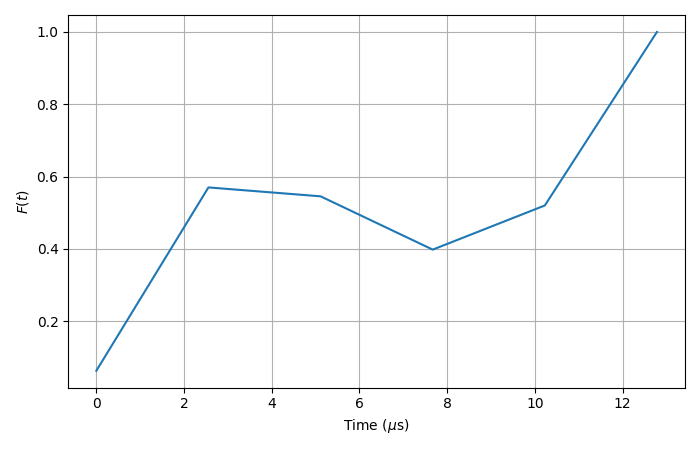}  
\caption{{\it{ Evolution of three Rydberg atoms with optimized parameters to reach the ground state of the deuteron Hamiltonian $H^{(1)}_{3,{\rm opt}}$ (Eq.~(\ref{eq:H3opt})). 
The top panel displays the probability density
evolution, and the bottom one the corresponding energy. The dashed line corresponds to the exact deuteron ground state energy. The Rydberg Hamiltonian for 2 atoms being symmetric in the exchange of the atoms, the states $|01\rangle$ and $|10\rangle$ are degenerated. The last Rydberg state has a $99.999\%$ fidelity with the deuteron ground state, represented by colored dots. 
}}}\label{fig:H3Simu_opt}
\end{center}
\end{figure}

The frequency optimization performed in this work is carried out classically using BFGS algorithm on a classical computer. This
should be regarded as a proof-of-principle demonstration that the corresponding control parameters exist, rather than as a requirement
for their experimental determination. For fixed atomic geometry, the same optimization can in principle be performed in a hardware-in-the-
loop fashion, with the Rydberg processor providing the measured objective function (or its gradient) and the control frequencies being
updated iteratively. The principal-angle analysis can then be used as a classical preprocessing step to identify a favorable encoding
orientation before the hardware optimization.

\section{Conclusion}
\label{sec:conclusion}

We have introduced a geometry-independent framework for quantifying and optimizing the compatibility between a target Hamiltonian and the native control resources of an analog quantum 
simulator, based on the principal angles between the target's spectral algebra $\mathcal A(H)$ and the operator space $\mathcal K$ reachable under the device's control Hamiltonians. Because every encoding related to $H$ by a unitary
transformation shares the same physical spectrum, this criterion turns encoding selection into a well-posed, smooth optimization over the unitary orbit of $H$, solved here by Riemannian gradient descent on $U(d)$. We implement this optimization using Riemannian gradient descent on the unitary group, providing a systematic procedure for selecting representations that are 
better aligned with the available control resources.

Applied to the Hamming-weight-one deuteron Hamiltonian on a Rydberg-atom analog processor, the framework reveals a clear distinction between algebra-level compatibility and state-level 
reachability. For $N=3$ and $N=4$, the standard-binary encoding is strongly misaligned with the geometry-independent Krylov space, whereas principal-angle optimization can bring the complete
target spectral algebra into that space. The subsequent dynamical optimization confirms that this algebraic improvement translates into essentially exact ground-state preparation. 
For $N=5,\ldots,8$, in contrast, complete algebra containment is not achieved even after encoding optimization. Nevertheless, high-fidelity ground-state preparation remains possible 
when the atomic geometry and dynamical controls are optimized jointly. Thus, containment of the full spectral algebra in $\mathcal K$ is a strong compatibility criterion, but it is not 
a necessary condition for preparing a particular target state.

The fixed-geometry calculations demonstrate the practical value of optimizing the encoding independently of the physical layout. For the standard-binary representation, the achievable 
fidelity varies strongly among the geometries considered, with near-zero fidelities occurring for several system sizes and geometries. A single principal-angle-optimized encoding, determined 
before the geometry is specified and then reused unchanged, substantially reduces this sensitivity and yields high fidelities across all five tested geometries. The remaining geometry dependence 
is expected: the geometry-independent optimization does not include the interaction Hamiltonian $H_0(\mathbf r)$ and therefore cannot, by construction, optimize the compatibility with 
every possible physical layout. The encoding optimization should consequently be viewed as a preprocessing step that improves the representation before the remaining geometry- and pulse-dependent control problem is solved.

These results point to a practical design principle that is particularly relevant given the current state of analog quantum hardware. Present-day neutral-atom and trapped-ion analog processors
implement, natively and with high fidelity, Ising-type two-body interactions driven by global control fields -- exactly the control structure assumed throughout this work. 
Engineering local, qubit-addressed control or realizing effective $XY$ (flip-flop) interactions generally requires additional experimental resources beyond the global-control, Ising-type interactions that are naturally available in many Rydberg-atom platforms~\cite{Morgado2020,Browaeys2020}.
 In this landscape, the encoding of a target problem is one of the few remaining design freedoms that requires no 
new hardware capability: it is a classical, offline optimization performed once, before any pulse is applied, and it can convert a Hamiltonian that is only reachable under $XY$-type or locally 
addressed control into one that is reachable under the Ising-plus-global-control resources a given device already provides. This is precisely the situation encountered here -- the deuteron 
Hamiltonian's native hopping structure is of $XY$ type, while the Rydberg platform's native interaction is Ising -- and it is a situation that recurs whenever a target model's natural qubit representation 
does not match the interaction type a given analog platform implements natively. Where a full joint optimization of encoding, control pulses, and atomic geometry is computationally 
prohibitive, or where the geometry is constrained by the hardware layout, the geometry-independent reachability criterion introduced here offers a low-cost preprocessing step that can be 
evaluated before committing to a physical implementation.

More broadly, the results show that spectral equivalence alone is insufficient to characterize the suitability of an encoding for analog quantum simulation. Two unitarily equivalent representations can have identical spectra while exhibiting very different compatibility with the native operator space of the hardware. Principal-angle optimization makes this distinction explicit and provides a constructive way to exploit the encoding freedom. Its role is not to replace dynamical controllability analysis or device-specific pulse and geometry optimization, but to identify and remove representation-level incompatibilities before those more expensive optimizations are performed. This separation between algebraic encoding design and physical control optimization offers a systematic route toward more robust mappings of target Hamiltonians onto constrained analog quantum hardware.

Several extensions are natural. First, the present analysis is restricted to small system sizes for which dense-matrix representations of $\mathcal A(H)$ and $\mathcal K$ remain tractable; extending the
construction to exploit Pauli sparsity, symmetry, or tensor-network representations, as discussed in Sec.~\ref{sec:scaling}, would be required to reach the system sizes of experimental interest. Second,
the geometry-independent construction of $\mathcal K$ deliberately excludes the interaction term $H_0(\mathbf r)$; a criterion that incorporates $H_0$ at a fixed, hardware-specified geometry would
sharpen the diagnostic for exactly the fixed-geometry regime found here to be most fragile, and could explain the residual, geometry-dependent effects -- such as the comparative robustness of $N=8$ and the
persistent difficulty of the least-symmetric tested geometry -- that the present, geometry-independent criterion cannot by construction capture. Third, while this work optimizes a single smooth objective
built from all principal angles, a worst-case objective based on the smallest singular value, or a direct treatment of state-level rather than algebra-level reachability, may prove more efficient for
larger systems where full algebra containment is not achievable; we introduced the former in Sec.~II and leave a systematic comparison to future work. Finally, the same framework applies unchanged
to any analog platform describable by a control-affine Hamiltonian, and to any target model whose qubit encoding is not fixed a priori -- including other benchmark problems in nuclear and high-energy
physics currently being explored on analog and digital quantum hardware~\cite{BauerNatRevPhys2023,BauerPRXQuantum2023}.


\begin{acknowledgments}
The author is grateful to D. Lacroix and V. Soma for fruitful discussions. This work is part of HQI initiative (www.hqi.fr).
\end{acknowledgments}

\appendix

\section{Krylov space construction for $N=3$ and $N=4$}\label{app:Krylov}
\begin{widetext}
\paragraph{Krylov space for $N=3$.}

The initial state density matrix expands as
\[
\rho_0=\frac18\Big[\mathbbm 1+(Z_0+Z_1+Z_2)+(Z_0Z_1+Z_0Z_2+Z_1Z_2)+Z_0Z_1Z_2\Big].
\]

Applying the previous proposition, we get:

\emph{First level commutators:} 
Using \eqref{eq:level1} with $N=3$:
\begin{align*}
[S_x,\rho_0]&=-\frac i4\Big[S^{(3)}_{0,1,0}+S^{(3)}_{0,1,1}+S^{(3)}_{0,1,2}\Big]\\
&=-\frac i4\Big[(Y_0+Y_1+Y_2)+\!\!\sum_{i<j}\!(Y_iZ_j+Z_iY_j)+(Y_0Z_1Z_2+Z_0Y_1Z_2+Z_0Z_1Y_2)\Big],\\[4pt]
[S_y,\rho_0]&=\frac i4\Big[S^{(3)}_{1,0,0}+S^{(3)}_{1,0,1}+S^{(3)}_{1,0,2}\Big]\\
&=\frac i4\Big[(X_0+X_1+X_2)+\!\!\sum_{i<j}\!(X_iZ_j+Z_iX_j)+(X_0Z_1Z_2+Z_0X_1Z_2+Z_0Z_1X_2)\Big],\\[4pt]
[S_z,\rho_0]&=0.
\end{align*}
The weight-3 sector is already populated at zeroth order, through the term $Z_0Z_1Z_2$ present in $\rho_0$ itself; the first-order commutators above then rotate this seed into the weight-3 operator $Y_0Z_1Z_2+Z_0Y_1Z_2+Z_0Z_1Y_2$ (and its $S_y$-counterpart with $X\leftrightarrow Y$), confirming that this sector is active and consistent with point (ii) of Proposition~\ref{prop:krylov}.

\emph{Second level commutators:} Iterating $S_x$ on its own commutator, using $[X_i,Y_i]=2iZ_i$, $[X_i,Z_i]=-2iY_i$, and $[X_i,X_i]=0$:
\begin{align*}
[S_x,[S_x,\rho_0]]&=\tfrac12 S^{(3)}_{0,0,1}+S^{(3)}_{0,0,2}-S^{(3)}_{0,2,0}+\tfrac32S^{(3)}_{0,0,3}-S^{(3)}_{0,2,1}\\
&=\tfrac12(Z_0+Z_1+Z_2)+\!\sum_{i<j}\!Z_iZ_j-\!\sum_{i<j}\!Y_iY_j+\tfrac32 Z_0Z_1Z_2-(Y_0Y_1Z_2+Y_0Z_1Y_2+Z_0Y_1Y_2),
\end{align*}
and, by the mirror computation with $[Y_i,X_i]=-2iZ_i$, $[Y_i,Z_i]=2iX_i$,
\begin{align*}
[S_y,[S_y,\rho_0]]&=\tfrac12S^{(3)}_{0,0,1}-S^{(3)}_{2,0,0}+S^{(3)}_{0,0,2}+\tfrac32S^{(3)}_{0,0,3}-S^{(3)}_{2,0,1}\\
&=\tfrac12(Z_0+Z_1+Z_2)-\!\sum_{i<j}\!X_iX_j+\!\sum_{i<j}\!Z_iZ_j+\tfrac32Z_0Z_1Z_2-(X_0X_1Z_2+X_0Z_1X_2+Z_0X_1X_2).
\end{align*}
These two double commutators already inject the pure-$Y$, pure-$X$, pure-$Z$ weight-2 operators and three of the weight-3 mixed types. The remaining compositions ($S^{(3)}_{1,1,0}$ at weight 2, and $S^{(3)}_{3,0,0},S^{(3)}_{0,3,0},S^{(3)}_{2,1,0},S^{(3)}_{1,2,0},S^{(3)}_{1,1,1}$ at weight 3) are generated in the same way by the cross double commutators $[S_x,[S_y,\rho_0]]$, $[S_y,[S_x,\rho_0]]$ and by third-level commutators; no qualitatively new mechanism appears, and no operator of weight $>3$ can ever be produced (point (i) above). The chain saturates once every composition of every weight $w\le 3$ has appeared.

The resulting Krylov space is therefore, organized by weight,
\begin{widetext}
\begin{eqnarray*}
&&{\cal K}_m^{\rm ctrl}(\rho_0)=\underbrace{\big[\,S_x,\,S_y,\,S_z\,\big]}_{w=1,\ \dim 3}\\
&\cup&
\underbrace{\big[\,X_0X_1{+}X_0X_2{+}X_1X_2,\ \ Y_0Y_1{+}\ldots,\ \ Z_0Z_1{+}\ldots,
\textstyle\sum_{i<j}(X_iY_j{+}Y_iX_j),\ \textstyle\sum_{i<j}(X_iZ_j{+}Z_iX_j),\ \textstyle\sum_{i<j}(Y_iZ_j{+}Z_iY_j)\,\big]}_{w=2,\ \dim 6}\\
&\cup&\big[\,X_0X_1X_2,\ Y_0Y_1Y_2,\ Z_0Z_1Z_2,
X_0X_1Y_2{+}X_0Y_1X_2{+}Y_0X_1X_2,\ \ (\text{and the 5 analogous $2{+}1$ mixed sums}),\\
&&\underbrace{X_0Y_1Z_2{+}X_0Z_1Y_2{+}Y_0X_1Z_2{+}Y_0Z_1X_2{+}Z_0X_1Y_2{+}Z_0Y_1X_2\,}_{w=3,\ \dim 10}\big]
\end{eqnarray*}
\end{widetext}
i.e.\ exactly one symmetrized operator $S^{(3)}_{n_X,n_Y,n_Z}$ for every composition of $w=1,2,3$ into three non-negative parts: $3+6+10=19$ operators, plus the identity, for a total dimension of 20 (19 without the identity).

As in the $N=2$ case, any component along an antisymmetric combination (e.g.\ $X_0X_1-X_1X_2$, or any sum not invariant under permutation of the three atoms) is inaccessible from $\rho_0$ under these controls.

\paragraph{Krylov space for $N=4$.}

The initial state now expands as
\[
\rho_0=\frac1{16}\Big[\mathbbm 1+\sum_iZ_i+\sum_{i<j}Z_iZ_j+\sum_{i<j<k}Z_iZ_jZ_k+Z_0Z_1Z_2Z_3\Big].
\]
The same first-level formula \eqref{eq:level1} gives
\begin{align*}
[S_x,\rho_0]&=-\frac i8\Big[S^{(4)}_{0,1,0}+S^{(4)}_{0,1,1}+S^{(4)}_{0,1,2}+S^{(4)}_{0,1,3}\Big],\\
[S_y,\rho_0]&=\frac i8\Big[S^{(4)}_{1,0,0}+S^{(4)}_{1,0,1}+S^{(4)}_{1,0,2}+S^{(4)}_{1,0,3}\Big],\qquad\\
 [S_z,\rho_0]&=0,
\end{align*}
where now the top term $S^{(4)}_{0,1,3}=\sum_{i}Y_i\!\!\prod_{k\ne i}\!Z_k$ (one $Y$, three $Z$'s, summed over the position of the $Y$ among all four atoms) already reaches the maximal weight $w=N=4$, coming directly from the $Z_0Z_1Z_2Z_3$ term of $\rho_0$ --  the direct four-atom analogue of what happened for $N=3$.

Exactly the same mechanism as above then applies within each weight sector: point (i) forbids any change of weight, and iterating $S_x,S_y$ (double and triple commutators, entirely analogous to the $[S_x,[S_x,\rho_0]]$, $[S_y,[S_y,\rho_0]]$ computations done explicitly for $N=3$) fills in every composition $(n_X,n_Y,n_Z)$ at each weight $w=1,2,3,4$, with no further growth once all compositions of weight $\le 4$ are present. The Krylov space is thus
\begin{widetext}
\[
{\cal K}_m^{\rm ctrl}(\rho_0)=\Big\{\,S^{(4)}_{n_X,n_Y,n_Z}\ \Big|\ n_X,n_Y,n_Z\ge0,\ 1\le n_X+n_Y+n_Z\le4\,\Big\},
\]
organized by weight as:
\begin{itemize}
\item $w=1$ (dim 3): $S_x,\,S_y,\,S_z$.
\item $w=2$ (dim 6): $\sum_{i<j}X_iX_j,\ \sum_{i<j}Y_iY_j,\ \sum_{i<j}Z_iZ_j,\ \sum_{i<j}(X_iY_j{+}Y_iX_j),\ \sum_{i<j}(X_iZ_j{+}Z_iX_j),\ \sum_{i<j}(Y_iZ_j{+}Z_iY_j)$, each summed over the $\binom42=6$ pairs.
\item $w=3$ (dim 10): the ten compositions of $3$ into $(n_X,n_Y,n_Z)$ --  $(3,0,0),(0,3,0),(0,0,3),(2,1,0),(2,0,1),(1,2,0),(0,2,1),(1,0,2),(0,1,2),(1,1,1)$ --  each symmetrized over the $\binom43=4$ triples of active atoms (and, for $(1,1,1)$, over the $3!=6$ letter assignments within each triple).
\item $w=4$ (dim 15): the fifteen compositions of $4$ into $(n_X,n_Y,n_Z)$ --  $(4,0,0),(0,4,0),(0,0,4)$; $(3,1,0),(3,0,1),(1,3,0),(0,3,1),(1,0,3),(0,1,3)$; $(2,2,0),(2,0,2),(0,2,2)$; $(2,1,1),(1,2,1),(1,1,2)$ --  each now involving all four atoms, e.g. $S^{(4)}_{4,0,0}=X_0X_1X_2X_3$, $S^{(4)}_{3,1,0}=X_0X_1X_2Y_3+X_0X_1Y_2X_3+X_0Y_1X_2X_3+Y_0X_1X_2X_3$, up to $S^{(4)}_{2,1,1}$ which sums all $4!/(2!1!1!)=12$ distinct placements of $\{X,X,Y,Z\}$ on the four atoms.
\end{itemize}
\end{widetext}

The dimension count is $3+6+10+15=34$, plus the identity, giving a total Krylov space dimension of 35 (34 without the identity).

As before, only permutation-symmetric combinations are reachable: any antisymmetric combination of same-weight strings (e.g.\ $X_0X_1X_2X_3$ terms built with a relative minus sign under atom exchange) lies outside ${\cal K}_m^{\rm ctrl}(\rho_0)$ and constitutes an obstruction to reachability of a target state from the fully polarized ground state under global (collective) Rydberg control.
\end{widetext}



\begin{thebibliography}{100}
\bibitem{Surace2020} F. Surace et al., {\it{Lattice Gauge Theories and String Dynamics in Rydberg Atom Quantum Simulators}}, Physical Review X 10, 021041 (2020)
\bibitem{Zhu2024} Z. H. Zhu et al.,   {\it{Probing false vacuum decay on a cold-atom gauge-theory quantum simulator}}, arXiv:2411.12565v1 [cond-mat.quant-gas]
\bibitem{Morgado2011} M. Morgado et al.,  {\it{Quantum simulation and computing with Rydberg-interacting qubits}}, arXiv:2011.03031v2 [quant-ph] 3
\bibitem{Georgescu2014} I. M. Georgescu, S. Ashhab, and F. Nori, {\it{Quantum simulation}}, Rev. Mod. Phys. 86, 153 (2014).
\bibitem{Zeier2011} R. Zeier and T. Schulte-Herbr{\"u}ggen, {\it{Symmetry principles in quantum systems theory}}, J. Math. Phys. 52, 113510 (2011).
\bibitem{DAlessandro2021} D. D'Alessandro and J. T. Hartwig, {\it{Dynamical decomposition of bilinear control systems subject to symmetries}}, J. Dyn. Control Syst. 27, 1 (2021).
\bibitem{DAlessandro2020} D. D'Alessandro, {\it{Topological properties of reachable sets and the control of quantum bits}}, Systems and Control Letters,
Volume 41, Issue 3, 2000, 213-221 (2000)
\bibitem{Wiersema2024} R. Wiersema, E. K{\"o}kc{\"u}, A. F. Kemper, and B. N. Bakalov, {\it{Classification of dynamical Lie algebras for translation-invariant 2-local spin systems in one dimension}}, npj Quantum Inf. 10, 110 (2024).
\bibitem{Khaneja2005} N. Khaneja, T. Reiss, C. Kehlet, T. Schulte-Herbr{\"u}ggen, and S. J. Glaser, {\it{Optimal control of coupled spin dynamics: design of NMR pulse sequences by gradient ascent algorithms}}, J. Magn. Reson. 172, 296 (2005).
\bibitem{BauerNatRevPhys2023} C. W. Bauer, Z. Davoudi, N. Klco, and M. J. Savage, {\it{Quantum simulation of fundamental particles and forces}}, Nat. Rev. Phys. 5, 420 (2023).
\bibitem{BauerPRXQuantum2023} C. W. Bauer et al., {\it{Quantum simulation for high-energy physics}}, PRX Quantum 4, 027001 (2023).
\bibitem{Dimitrescu2018}E. F. Dumitrescu et al. {\it{Cloud Quantum Computing of an Atomic Nucleus}}, Physical Review Letters 120, 210501 (2018)
\bibitem{Siwach2021} P. Siwach and P. Arumugam, {\it{Quantum simulation of nuclear Hamiltonian with a generalized transformation for Gray code encoding}}, Phys. Rev. C 104, 034301 (2021).
\bibitem{Absil2008} P.-A. Absil, R. Mahony, and R. Sepulchre, l. {\it{Optimization Algorithms on Matrix Manifolds}}
(Princeton University Press, Princeton, NJ, 2008).
\bibitem{SchulteHerbrueggen2008} T. Schulte-Herbr\"uggen, S. J. Glaser, G. Dirr, and U. Helmke,
l. {\it{Gradient Flows for Optimisation and Quantum Control: Foundations and Applications,}}
arXiv:0802.4195 [quant-ph] (2008).
 \bibitem{Nocedal2006} J. Nocedal and S. J. Wright, l. {\it{Numerical Optimization}}, 2nd ed. (Springer, New York, 2006).
\bibitem{Wiersema2023} R. Wiersema and N. Killoran, l. {\it{Optimizing quantum circuits with Riemannian gradient flow,}}
Phys. Rev. A \textbf{107}, 062421 (2023)
\bibitem{Browaeys2020} A. Browaeys and T. Lahaye, {\it{Many-body physics with individually controlled Rydberg atoms}}, Nat. Phys. 16, 132 (2020).
\bibitem{Morgado2020}M. Morgado et al.,  ``Quantum simulation and computing with Rydberg-interacting qubits",
 AVS Quantum Science, 2020, https://api.semanticscholar.org/CorpusID:234771773
\end{thebibliography}
\end{document}